\documentclass[letterpaper,11pt]{article}
\usepackage[margin=1in]{geometry}

\usepackage{graphicx}
\usepackage{amsfonts}
\usepackage{mathtools, amsthm, amssymb}
\usepackage{thmtools}
\usepackage[hidelinks]{hyperref}
\usepackage{cleveref}
\usepackage{todonotes}
\usepackage{cases}
\usepackage{todonotes}

\declaretheorem{theorem}
\declaretheorem[sibling=theorem]{lemma}
\declaretheorem[sibling=theorem]{corollary}
\declaretheorem{problem}
\declaretheorem[sibling=theorem]{definition}
\declaretheorem[sibling=theorem]{proposition}

\theoremstyle{remark}
\declaretheorem[sibling=theorem]{example}

\newcommand{\nat}{\mathbb{N}}
\newcommand{\intg}{\mathbb{Z}}
\newcommand{\rel}{\mathbb{R}}
\newcommand{\rat}{\mathbb{Q}}

\newcommand{\Acal}{\mathcal{A}}
\newcommand{\Bcal}{\mathcal{B}}

\newcommand{\Lcal}{\mathcal{L}}
\newcommand{\Mcal}{\mathcal{M}}

\newcommand{\Pcal}{\mathcal{P}}

\newcommand{\Scal}{\mathcal{S}}
\newcommand{\Tcal}{\mathcal{T}}

\newcommand{\seq}[1]{(#1)_{n=0}^{\infty}}
\newcommand{\st}{\colon}

\newcommand{\zerovec}{\mathbf{0}}

\title{On the Decidability of \\Presburger Arithmetic Expanded with
  Powers}
  
\usepackage{authblk}
\author[1]{Toghrul Karimov}
\author[2]{Florian Luca}
\author[1]{Joris Nieuwveld}
\author[1]{Jo{\"e}l Ouaknine}
\author[3]{James Worrell}
\affil[1]{Max Planck Institute for Software Systems, Saarland
  Informatics Campus, Germany}
\affil[2]{Mathematics Division, Stellenbosch University, Stellenbosch, South Africa}
\affil[3]{Oxford University, Department of Computer Science, Oxford, United Kingdom}
\date{}

\begin{document}
	\maketitle
	\begin{abstract}
          We prove that for any integers $\alpha, \beta > 1$, the
          existential fragment of the first-order theory of the
          structure
          $\langle \intg; 0,1,<, +, \alpha^{\nat},
          \beta^{\nat}\rangle$ is decidable (where $\alpha^{\nat}$ is
          the set of positive integer powers of $\alpha$, and likewise
          for $\beta^{\nat}$).  On the other hand, we show  by way of
          hardness that
          decidability of the existential fragment of the theory of
          $\langle \nat; 0,1, <, +, x\mapsto \alpha^x, x \mapsto
          \beta^x\rangle$ for any multiplicatively independent
          $\alpha,\beta > 1$ would lead to mathematical breakthroughs
          regarding base-$\alpha$ and base-$\beta$ expansions of
          certain transcendental numbers.
          Finally, modifying the original proof of Hieronymi and Schulz we show that for any multiplicatively independent $\alpha,\beta>1$, it is undecidable whether a given formula with at most 3 alternating blocks of quantifiers holds in $\langle \nat; 0, 1, <, +, \alpha^\nat, \beta^\nat \rangle$.
	\end{abstract}

\section{Introduction}

\emph{Presburger arithmetic}, the first-order theory of the integers with
addition and order, has been an object of study for nearly a
century. Its decidability was first established by Presburger in 1929
via a quantifier-elimination procedure~\cite{presburger-og}; yet
Presburger arithmetic remains to this day a topic of active research
owing, among others, to its deep connections to automata theory and formal
languages (see, e.g., the survey~\cite{haase-presburger-survival}) as
well as symbolic dynamics and combinatorics on words (see, e.g., the
excellent recent text~\cite{shallit22}).\footnote{It is interesting to note that
  the computational complexity of quantifier elimination itself
  remains of contemporary interest: see, e.g., \cite{haase24}.}  

Another rich line of inquiry has consisted in investigating
\emph{expansions} of Presburger arithmetic, i.e., theories obtained by
augmenting Presburger arithmetic with particular predicates or
functions. Here one must proceed with care: adding, for example, the
multiplication function $\times : \intg^2 \rightarrow \intg$ (or even
simply the `squaring' function, from which multiplication is easily
recovered) to Presburger arithmetic immediately results in
undecidability, thanks to G\"odel's incompleteness
theorem~\cite{godel31}. In fact, even the existential fragment of the
first-order theory of $\langle \intg; 0,1,<,+,\times\rangle$ is
undecidable, as shown by Matiyasevich in his negative solution of
Hilbert's 10th problem (see~\cite{mat93}). Nevertheless, many
decidable expansions of Presburger arithmetic have been discovered and
studied (see, for instance, the
survey~\cite{bes2002survey}). Decidability is usually established in one
of two ways: either via quantifier elimination, along the lines of
Presburger's original approach, or through automata-theoretic means,
where integers are encoded in a given base as strings of digits, which
are in turn manipulated by automata.

Before giving examples of such expansions, let us introduce some
notation. For a fixed integer $\alpha \geq 2$, we denote by $\alpha^\nat$
the set $\{ \alpha^n : n \in \nat\}$ of all positive powers of
$\alpha$, and by $\alpha^x$ the function $n \mapsto \alpha^n$
that takes a positive integer argument $n$ to $\alpha^n$. We also
write $V_\alpha(n)$ to represent the function taking $n$ to the
largest power of $\alpha$ that divides $n$ (thus, for example, $V_2(24) = 8$).

Using automata theory, B\"uchi showed that, for any $\alpha$, the
first-order theory of $\langle \intg; 0,1,<,+,V_{\alpha}\rangle$ is
decidable~\cite{buchi1960weak}. Villemaire however proved that, for
multiplicatively independent $\alpha$ and $\beta$, the first-order
theory of $\langle \intg; 0,1,<,+,V_{\alpha},V_{\beta}\rangle$ is
undecidable~\cite{villemaire92}.  Sem\"enov used quantifier
elimination to show that, for any `effectively sparse' predicate
$P \subset \intg$, the first-order theory of
$\langle \intg; 0,1,<,+,P\rangle$ is decidable. Examples of sparse
predicates include the sets of powers $\alpha^\nat$ as well as the set
of factorial numbers $\{n! : n \in\nat\}$.\footnote{The complexity of
  expansions of Presburger arithmetic by a power predicate
  $\alpha^\nat$ or a powering function $\alpha^x$ was very recently
  investigated~\cite{benedikt23}.}  The question of whether
decidability could however be maintained with the addition of
\emph{two} (or more) power predicates goes back to the 1980s; it was
finally answered in the negative in a recent paper of Hieronymi and
Schulz~\cite{hieronymi-schulz}.

Note that automata-theoretic techniques work well when all numbers in play can be
represented over a common base. But unfortunately, for
multiplicatively independent $\alpha$ and $\beta$ (such as $2$ and~$3$), this is not the case: powers of $2$, for example, have a very
regular structure in base $2$ but not in base~$3$, and
vice-versa. Moreover, multiplicatively independent
power predicates enable one to formulate non-trivial number-theoretic
assertions about integers, such as the fact that there are only
finitely many powers of $2$ and powers of $3$ that are no farther than
$10$ apart, say. Such an assertion can in fact already be formulated in the
first-order theory of $\langle \intg; 0,1,<,2^\nat,3^\nat\rangle$
(noting that addition has been removed); the decidability of this
theory is non-trivial, and was established by
Sem\"enov~\cite{semenov1980}. Very recently, the monadic second-order theory of
$\langle \intg; 0,1,<,2^\nat,3^\nat\rangle$ was also shown
decidable~\cite{MSO24}. 

Hieronymi and Schulz's undecidability result is quite intricate. The
standard approach would have been to show that multiplication is
definable in $\langle \intg; 0,1,<,+,\alpha^\nat,\beta^\nat\rangle$, but
unfortunately, this is provably not the
case~\cite{schulz2023undefinability}. The undecidability construction
in~\cite{hieronymi-schulz} makes use of three quantifier alternations
(i.e., four blocks of quantifiers of alternating polarity). This
naturally raises the question of whether weaker fragments might be
decidable.
In~\cite[Section~5]{hieronymi-schulz},
Hieronymi and Schulz in fact conjecture that the \emph{existential} fragment
of $\langle \intg; 0,1,<,+,\alpha^\nat,\beta^\nat\rangle$ is decidable
subject to certain number-theoretic effectiveness assumptions.

\textbf{Our main contribution is the following:}

\begin{theorem}
  \label{thm::main-decidability}
There is an algorithm that, given integers
                $\alpha,\beta > 1$
                and an existential formula $\varphi$ of $\langle
                \intg; 0,1,<,+,\alpha^\nat,\beta^\nat\rangle$,
                decides whether $\varphi$ is true or not.
	\end{theorem}

        As noted above, automata-theoretic techniques appear
        inadequate to establish such statements. We make use instead
        of mathematical tools from Diophantine approximation and
        transcendental number theory, in particular Baker's theorem
        on linear forms in logarithms, in a manner similar
        to~\cite{bertok2017number,stewart1980representation}.

As a secondary contribution, we provide a shorter proof of Hieronymi and Schulz's
undecidability result, requiring only two quantifier alternations
(rather than three); this is presented in
Section~\ref{sec::undec}. 

Finally, we also investigate the existential fragment of $\langle
                \nat; 0,1,<,+,\alpha^x,\beta^x\rangle$, in which the
                power predicates have been replaced by powering
                functions.\footnote{We have switched the domain from
                  $\intg$ to $\nat$; this is entirely benign, as
                  the order relation is available to us, and was
                  carried out chiefly so as not to have to separately redefine
                  the meaning of the powering functions over negative
                  entries.} 
                We have not been able to establish either decidability
                or undecidability; however, we prove the following by
                way of hardness:

	\begin{theorem}
		\label{thm::main-hardness}
		Let $\alpha, \beta > 1$ be multiplicatively independent integers.
		Write $\seq{A_n}$ for the base-$\beta$ expansion of $\log_\beta(\alpha)$ and $\seq{B_n}$ for the base-$\alpha$ expansion of $\log_\alpha(\beta)$.
		Suppose that the existential fragment of $\langle
                \nat; 0,1,<,+,\alpha^x,\beta^x\rangle$ is decidable.
		Then the following are in turn decidable:
		\begin{enumerate}
			\item[(A)] Whether a given pattern appears in $\seq{A_n}$.
			\item[(B)] Whether a given pattern appears at some index simultaneously in $\seq{A_n}$ and $\seq{B_n}$.
			\item[(C)] Whether a given pattern appears in $(A_{\alpha^n})_{n=0}^\infty$.
		\end{enumerate}
              \end{theorem}

To place Theorem~\ref{thm::main-hardness} in context, consider the
case of $\alpha=2$ and $\beta=3$. 
The constant $\log_3(2)$ is a transcendental
number that is widely conjectured to be \emph{normal}
(and thus in base $3$, every length-$l$ pattern should appear
within $\seq{A_n}$ with density $3^{-l}$). A fortiori, this would
entail that the answer to the first query is always positive. However,
normality on its own is not sufficient to settle either of
the other two queries.

	\section{Mathematical background}
	\label{sec::prelims}

	We denote by $\zerovec$ a (column) vector of all zeros whose dimension will be clear from the context.
	We will occasionally write $d$-dimensional column vectors in the form $(x_1,\ldots,x_d)$.
	For vectors $\mathbf{x} = (x_1,\ldots,x_d)$, $\mathbf{y} = (y_1,\ldots,y_d)$, and a relation $\sim$, we write $\mathbf{x} \sim \mathbf{y}$ as a shorthand for $x_i \sim y_i$ for all~$i$.
	For a ring~$R$, by an \emph{$R$-linear form} we mean a function of the form
	$
	h(x_1,\ldots,x_l) \coloneqq c_1x_1+\cdots+c_lx_l
	$
	where $c_i \in R$ for all $i$.
	We say that $\alpha, \beta \ne 0$ are \emph{multiplicatively independent} if for all $n_1, n_2 \in \nat$, $\alpha^{n_1} = \beta^{n_2}$ implies $n_1= n_2= 0$.

	\subsection{Logical theories}

	A \emph{structure} $\mathbb{M}$ consists of a universe $U$, constants $c_1,\ldots,c_k \in U$, predicates $P_1,\ldots,P_l$ where each $P_i \subseteq U^{\mu(i)}$ for some $\mu(i) \ge 1$, and functions $f_1,\ldots,f_m$ where each $f_i$ has the type $f_i \st U^{\delta(i)} \to U$ for some $\delta(i) \ge 1$.
	By the \emph{language} of the structure $\mathbb{M}$ we mean the set of all well-formed first-order formulas constructed from symbols denoting  the constants $c_1,\ldots, c_k$, predicates $P_1,\ldots,P_l$, and functions $f_1,\ldots, f_m$, as well as the symbols $\forall, \exists, \land, \lor, \lnot, =$.
	We will additionally write $x \in P$ for a unary predicate $P$ to mean $P(x)$.
	A \emph{term} is a well-formed expression constructed from constant, function, and variable symbols.
	Terms represent elements of the universe.
	A \emph{theory} is simply a set of sentences, i.e.\ formulas without free variables.
	The theory of the structure $\mathbb{M}$ is the set of all sentences in the language of $\mathbb{M}$ that are true in $\mathbb{M}$.
	A formula is \emph{existential} if it is of the form $\exists x_1 \cdots \exists x_m \st \varphi(x_1,\ldots,x_m)$ for $\varphi$ quantifier-free.
	The \emph{existential fragment} of a theory $\Tcal$, which itself is a theory, is the set of all existential formulas belonging to~$\Tcal$.
	Finally, a theory $\Tcal$ is \emph{decidable} if there exists an algorithm that takes a sentence $\varphi$ and decides whether $\varphi \in \Tcal$.

	For a positive integer $x$, denote by $x^{\nat}$ the unary predicate $\{x^n \st n \in\nat\}$.
	Let $\alpha,\beta > 1$.
	We will be working with the following structures and their theories.
	\begin{itemize}
		\item Let $\mathbb{M}_1  = \langle \intg; 0,1,<,+,\alpha^{\nat},\beta^{\nat}\rangle$.
		We will denote the language of this structure by $\Lcal_{\alpha,\beta}$ and its theory by $\mathcal{PA}(\alpha^\nat, \beta^\nat)$; in case $\alpha = \beta$, we will write $\mathcal{PA}(\alpha^\nat)$ for the latter.
		Observe that using the constants $0,1$ and addition, we can obtain any constant $c \in \nat$.
		On the other hand, $-c$ for $c>0$ can be accessed via the relation $x+c = 0$.
		In fact, for any $\intg$-linear form $h$ over $k$ variables,
		we can express $h(x_1,\ldots,x_k) = 0$ in the language $\Lcal_{\alpha,\beta}$ as $s(x_1,\ldots,x_k) = t(x_1,\ldots,x_k)$ where $s,t$ are $\intg$-linear forms with non-negative integer constants.
		Therefore, every atomic formula in $\Lcal_{\alpha,\beta}$ is equivalent to either $t \sim 0$ or $t \in \gamma^{\nat}$, for ${\sim} \in \{>, =\}$,  $\gamma \in \{\alpha,\beta\}$ and $t$ an integer linear combination of integer constants and variables.
		\item Let $\mathbb{M}_2 = \langle \nat; 0,1,<, +,  x \mapsto \alpha^x,  x \mapsto \beta^x \rangle$.
		That is, for $\gamma \in \{\alpha,\beta\}$, instead of the predicate~$\gamma^{\nat}$ we have the function that maps $x$ to $\gamma^x$.
		We write $\mathcal{PA}(\alpha^x,\beta^x)$ for the theory of $\mathbb{M}_2$.
		Note that the universe of $\mathbb{M}_2$ is $\nat$ as opposed to $\intg$.
		This is to ensure that the functions are total and map into the universe of the structure.
		For $\gamma \in \{\alpha,\beta\}$, we can express $x \in \gamma^\nat$ as $\exists z \st \gamma^z = x$.
		Therefore, if we can decide (the existential fragment of) $\mathcal{PA}(\alpha^x,\beta^x)$ then we can also decide (the existential fragment of) $\mathcal{PA}(\alpha^\nat,\beta^\nat)$.
	\end{itemize}

	A set $X \subseteq U^d$ is \emph{definable} in a structure $\mathbb{M}$ if there exists a formula $\varphi$ in the language of $\mathbb{M}$ with~$d$ free variables such that for all $x_1,\ldots,x_d \in U$, $\varphi(x_1,\ldots,x_d)$ is true if and only if $(x_1,\ldots,x_d) \in X$.
	A set $X \subseteq \intg^{d}$ is \emph{semilinear} if it is definable in the structure $\mathbb{M}_0 \coloneqq \langle \intg; 0, 1, <, + \rangle$.
	We write~$\Lcal$ for the language of $\mathbb{M}$, and $\mathcal{PA}$ for its theory.
	By the result of Presburger that the theory of $\mathbb{M}_0$ admits \emph{quantifier elimination} if we allow a divisibility predicate \cite{haase-presburger-survival}, such $X$ can be defined by a formula of the form
	\begin{equation}
		\label{eq::pres-qe-first}
		\bigvee_{i \in I}
		\biggl(
		\,
		\bigwedge_{j=J_i} t_{j}(x_1,\ldots,x_d) \equiv 0 \bmod D_{j} \:\:\land\:\:  \bigwedge_{k \in K_j} h_{k}(x_1,\ldots,x_d) \sim_{k} c_k
		\biggr)
	\end{equation}
	where $D_j \ge 1$ and each $t_{j}, h_{j}$ is a $\rat$-linear form, $c_k \in \intg$, and ${\sim}_{k} \in \{>,=\}$.

	\subsection{Number theory}
	Let $x \in \intg$ and $p \in \nat$ be a prime.
	Then the \emph{$p$-adic valuation} of $x$, denoted $\nu_p(x)$, is the largest integer $n$ such that $p^n$ divides $x$, whereas $p^{n+1}$ does not.
    By convention, $\nu_p(0) = +\infty$.
	For integers $x,y$ and a prime $p$ we have $\nu_p(x+y) \ge \min \{\nu_p(x),\nu_p(y)\}$.
	Let $z = \frac{a}{b} \in \rat$ be non-zero with $\gcd(a,b)=1$.
	Then the \emph{(absolute logarithmic) height} of~$z$ is $h(z) \coloneqq \max \{\log |a|, \log |b|\}$.
	For $z_1,\ldots,z_k \in \rat_{\ne 0}$ we have that $h(1/z_i) = h(z_i)$,
	\[
	h(z_1 + \cdots + z_k) \le h(z_1)+\cdots+h(z_k) +\log(k)
	\] and
	\[
	h(z_1\cdots z_k) \le h(z_1) + \cdots + h(z_k).
	\]
	The following is a specialisation of Matveev's version~\cite{matveev2000explicit} of Baker's theorem on linear forms to rational numbers.

	\begin{theorem}
		\label{thm::baker}
		Suppose we are given $k \ge 0$, non-zero $\gamma_1,\ldots,\gamma_k \in \rat$, and $b_1,\ldots,b_k \in \intg$.
		Write $B = \max \{1, |b_1|,\ldots,|b_k|\}$, $A_i = \max \{h(\gamma_i), |\log(\gamma_i)|, 0.16\}$ for $1 \le i \le k$, and
		\[
		\Lambda = \gamma_1^{b_1}\cdots\gamma_k^{b_k} - 1.
		\]
		Then, assuming $\Lambda \ne 0$,
		\[
		\log |\Lambda| > -1.4 \cdot 30^{k+3} \cdot k^{4.5} \cdot (1+\log(kB))\cdot A_1\cdots A_k.
		\]
	\end{theorem}

	The following is a consequence of Kronecker's theorem in Diophantine approximation \cite{gonek-kronecker}.
	\begin{lemma}
		\label{thm::kronecker}
		Let $\alpha,\beta \in \nat_{>1}$ be multiplicatively independent and $I \subseteq \rel_{>0}$ be a non-empty open interval.
		Then there exist infinitely many $n_1,n_2 \in \nat$ such that $\alpha^{n_1}/\beta^{n_2} \in I$.
	\end{lemma}
	\begin{proof}
		By multiplicative independence, $\log_\beta(\alpha)$ is irrational.
		Write $\{x\}$ for the fractional part of~$x$.
		By Kronecker's theorem, $\seq{ \{n\log_\beta(\alpha)\} }$ is dense in $(0,1)$.
		That is,
		\[
		\{n_1\log_\beta(\alpha) - n_2 \st n_1,n_2\in\nat\} \cap (0,1)
		\] is dense in $(0,1)$.
		Equivalently, $\{\alpha^{n_1}/\beta^{n_2} \st n_1,n_2\in\nat\} \cap (1,\beta)$ is dense in $(1,\beta)$.
		It follows that $\{\alpha^{n_1}/\beta^{n_2} \st n_1,n_2\in\nat\}$ is dense in $(0, \infty)$.
	\end{proof}

	\section{Overview of the results}
	\label{sec::overview}

	Recall that our central problem is to decide, given $\alpha,\beta$ and
	an existential formula $\varphi \in \Lcal_{\alpha,\beta}$,
	whether $\varphi \in \mathcal{PA}(\alpha^\nat,\beta^\nat)$.
	As the first step in our decidability proof, we will reduce our main problem to the following.

	\begin{problem}
		\label{problem1}
		Given multiplicatively independent $\alpha, \beta \in \nat_{>1}$, $z_1,\ldots,z_l \in \{\alpha,\beta\}$, $r,s \ge 0$, $A \in \intg^{r \times l}$, $\mathbf{b} \in \intg^{r}$, $C \in \intg^{s \times l}$, and $\mathbf{d} \in \intg^s$, decide whether there exists $\mathbf{z} = (z_1^{n_1},\ldots,z_l^{n_l})$ such that $A\mathbf{z} > \mathbf{b}$ and $C\mathbf{z} = \mathbf{d}$.
	\end{problem}

	The reduction from $\mathcal{PA}(\alpha^{\nat},\beta^\nat)$ to Problem~\ref{problem1} is captured by the following lemma.
	The proof, given in \Cref{sec::first-step}, uses fairly standard arguments about Presburger arithmetic.

	\begin{lemma}
		\label{thm::pres-to-system-reduction}
		Let $\alpha,\beta \in \nat_{>1}$.
		\begin{itemize}
			\item[(a)] If $\alpha,\beta$ are multiplicatively dependent, then deciding the existential fragment of $\mathcal{PA}(\alpha^{\nat},\beta^{\nat})$ reduces to deciding the existential fragment of the theory of $\mathcal{PA}(\gamma^\nat)$ for some $\gamma\in\nat$.
			\item[(b)] If $\alpha, \beta$ are multiplicatively independent, then deciding the existential fragment of $\mathcal{PA}(\alpha^{\nat},\beta^{\nat})$ reduces to Problem~\ref{problem1}.
		\end{itemize}
	\end{lemma}
	Recall that the full theory $\mathcal{PA}(\gamma^\nat)$ is known to be decidable.
	Hence it remains to show decidability of Problem~\ref{problem1}, which we do in \Cref{sec::equations,sec::ineqs}.
	\begin{theorem}
		\label{thm::problem1-decidable}
		Problem~\ref{problem1} is decidable.
	\end{theorem}
	We approach Problem~\ref{problem1} by first studying how to solve systems of the form $C\mathbf{z} = \mathbf{d}$, i.e.\ the case where there are no inequalities.
	The following definition captures the structure of solutions of such systems.

	\begin{definition}
		A set $X \subseteq \nat^{l}$ belongs to the class $\mathfrak{A}$ if it can be written in the form
		\begin{equation}
			\label{eq::aclass-form}
			X = \bigcup_{i \in I} \: \bigcap_{j \in J_i} X_{j}
		\end{equation}
		where $I$ and $J_i$ for every $i \in I$ are finite, and each $X_{j}$ is either of the form
		\begin{equation}\label{eq::aclass-1}
			X_{j} = \big\{(n_1,\ldots,n_l) \in \nat^l \st n_{\mu(j)} = n_{\sigma(j)} + c_{j}\big\}
		\end{equation}
		or of the form
		\begin{equation}
			\label{eq::aclass-2}
			X_{j} = \big\{(n_1,\ldots,n_l) \in \nat^l  \st n_{\xi(j)} = b_{j}\big\}
		\end{equation}
		where $1 \le \xi(j), \mu(j), \sigma(j) \le l$ and $b_{j}, c_{j} \in \nat$.
	\end{definition}
	The sets belonging to $\mathfrak{A}$ are semilinear.
	Observe that every finite subset of $\nat^l$ belongs to $\mathfrak{A}$, and the class $\mathfrak{A}$ is closed under finite unions and intersections.
	In \Cref{sec::equations}, we will prove the following structure and effectiveness result about the system $C\mathbf{z} = \mathbf{d}$.
	Our main tool is Baker's theorem on linear forms, which is frequently used when solving Diophantine equations where the unknowns appear in the exponent position.

	\begin{restatable}{theorem}{equationsmain}
		\label{thm::equalities-main}
		Let $\alpha,\beta \in \nat_{>1}$ be multiplicatively independent and $z_1,\ldots,z_l \in \{\alpha,\beta\}$ for some $l \ge 1$.
		Further let $s \ge1$, $C \in \intg^{s \times l}$, $\mathbf{d} \in \intg^s$, and $\Scal \subseteq \nat^l$ be the set of solutions of $C \mathbf{z} = \mathbf{d}$, where $\mathbf{z} = (z_1^{n_1},\ldots,z_l^{n_l})$.
		Then $\Scal \in \mathfrak{A}$.
		Moreover, a representation of $\Scal$ in the form~\eqref{eq::aclass-form} can be effectively computed, with the additional property that $z_{\mu(j)} = z_{\sigma(j)}$ for every $X_{j}$ of the form~\eqref{eq::aclass-1}.
	\end{restatable}

	When proving \Cref{thm::equalities-main}, because the class $\mathfrak{A}$ is closed under intersections, it suffices to consider a single equality
	\begin{equation}
		\label{eq::helper-1536}
		c_1z_1^{n_1} + \cdots + c_lz_l^{n_l} = d
	\end{equation}
	where $c_1,\ldots,c_l\in \intg_{\ne 0}$, $d\in \intg$, $\alpha, \beta > 1$ are multiplicatively independent, and $z_1,\ldots,z_l \in \{\alpha,\beta\}$.
	We will show that the set $\Scal$ of solutions of \eqref{eq::helper-1536} belongs to $\mathfrak{A}$ and has an effectively computable representation.
	Let us further stipulate that $z_i = \alpha$ and $z_j = \beta$ for some $i, j$, and that no proper sub-sum of the left-hand side of~\eqref{eq::helper-1536} is zero.
	In this case,
	it can be shown that the set of solutions is finite and can be effectively computed; see the proof of \Cref{thm::equalities-big-thm}.\footnote{For convenience, we make additional assumptions on $z_1^{n_1},\ldots,z_l^{n_l}$ in the statement of \Cref{thm::equalities-big-thm}.
		A slightly modified proof can be given to show finiteness of solutions of \eqref{eq::helper-1536} only assuming that both $\alpha,\beta$ appear among $z_1,\ldots,z_l$ and requiring that all proper sub-sums of $c_1z_1^{n_1} + \cdots + c_lz_l^{n_l}$ be non-zero.}
	The idea is to use Baker's theorem on linear forms iteratively to bound the gaps between $n_1,\ldots,n_l$, which, in case $d\ne0$, will yield an upper bound on all of $n_1,\ldots,n_l$.
	If $d=0$, then we need an additional argument involving $p$-adic valuations.
	On the other hand, if $\bigcap_{j\in J_i} X_{j}$ is infinite for some $i$ in the representation of $\Scal$  in the form~\eqref{eq::aclass-form}, then some sub-sum of $c_1z_1^{n_1} + \cdots + c_lz_l^{n_l}$ must be zero at infinitely many points $(n_1,\ldots,n_l)$.

	\begin{example}
	     Consider the equation
	\begin{equation}
		\label{eq::helper-2052}
		15\cdot3^{n_1} - 5\cdot 3^{n_2} + 2^{n_3} = 8.
	\end{equation}
	The only proper sub-sum of the left-hand side that can be zero is $15\cdot3^{n_1} - 5\cdot 3^{n_2}$.
	We therefore have the infinitely many solutions
	\[
	X \coloneqq \big\{(n_1,n_2,n_3) \in \nat^3 \st n_2 = n_1 + 1 \land n_3 = 3\big\}.
	\]
	Now suppose no proper sub-sum is zero.
	Let us additionally stipulate that $3^{n_2} \ge  3^{n_1} \ge 2^{n_3}$.
	In this case, if $n_2 > n_1 + 1$, then the summand $5 \cdot 3^{n_2}$ becomes too large in magnitude: we have that $5 \cdot 3^{n_2} \ge 45 \cdot 3^{n_1}, 45 \cdot 2^{n_3}$ and hence~\eqref{eq::helper-2052} cannot hold.
	Therefore, we are left with the possibilities $n_2 = n_1$ and $n_2 = n_1 + 1$.
	If we substitute $n_2 = n_1$ into \eqref{eq::helper-2052}, we obtain $10 \cdot 3^{n_2} + 2^{n_3} = 8$, which does not have a solution.
	The substitution $n_2 = n_1 + 1$, meanwhile, is not permitted as $15\cdot3^{n_1} - 5\cdot 3^{n_2}$ becomes zero.

	Using the same argument as above, we can handle the case where $3^{n_1} \ge  3^{n_2} \ge 2^{n_3}$.
	The four remaining cases (e.g.\ $3^{n_1} \ge  2^{n_3} \ge 3^{n_2}$), however, require an iterated application of Baker's theorem on linear forms as in~\Cref{thm::equalities-big-thm} to bound the solutions.
	Checking all possible $(n_1,n_2,n_3)$ up to this bound, we obtain that the set of all solutions of~\eqref{eq::helper-2052} is $\{(0,3,7), (1,8,15)\} \cup X$.
	\end{example}

	Once we know how to solve systems of linear equations in powers of $\alpha$ and $\beta$, we discuss how we deal with inequalities.
	In \Cref{sec::ineqs}, we argue as follows.
	Consider a system $A\mathbf{z} > \mathbf{b}$ and $C\mathbf{z} = \mathbf{d}$ as in the statement of Problem~\ref{problem1}.
	Observing that $x > -c$ is equivalent to $x = -c+1 \lor \cdots \lor x =0 \lor x > 0$ for any variable $x$ and positive integer $c$, we rewrite our system in the form
	\[
	\bigvee_{k \in K} A_k\mathbf{z} > \mathbf{b}_k \:\land\: C_k\mathbf{z} = \mathbf{d}_k
	\]
	where $A_k \in \intg^{r \times l}, C_k \in \intg^{s \times l}, \mathbf{b}_k \in \intg^r, \mathbf{d}_k \in \intg^s$ and, importantly, $\mathbf{b}_k \ge \zerovec$ for all $k$.
	We can now solve each $A_k\mathbf{z} > \mathbf{b}_k \:\land\: C_k\mathbf{z} = \mathbf{d}_k$ separately.
	Denote by $\Scal_k$ the set of all $(n_1,\ldots,n_l) \in \nat^l$ that satisfy $C_k\mathbf{z} = \mathbf{d}_k$.
	By \Cref{thm::equalities-main}, apart from finitely many exceptional solutions which can be effectively computed, the set $\Scal_k$ is defined by equations of the form either $n_a = n_b + c$ or $n_a = c$, where $1 \le a,b \le l$, $c \in \nat$, and $z_a = z_b$ in the former case.
	We can use each such equation as a substitution rule to eliminate the variable $z_a$; see~the proof of \Cref{thm::problem1-decidable} on \Cpageref{proof-of-problem1-decifable} for the exact procedure.
	In the end, we construct  $\widetilde{A}_k \in \intg^{r \times l}$ such that $A_k\mathbf{z} > \mathbf{b}_k \:\land\: C_k\mathbf{z} = \mathbf{d}_k$ has a solution if and only if $\widetilde{A}_k\mathbf{z} > \mathbf{b}_k$ has a solution.

	It remains to show how to solve the system $\widetilde{A}_k\mathbf{z} > \mathbf{b}_k$.
	To do this, we first argue that  $\widetilde{A}_k\mathbf{z} > \mathbf{b}_k$ has a solution if and only if $\widetilde{A}_k\mathbf{z} > \zerovec$ has a solution.
	Next, using a form of Fourier-Motzkin elimination, we reduce solving the latter system to solving systems of the form
	\begin{equation}
		\label{eq::overview-system}
		\begin{cases}
			\vspace*{0.15cm} 	\frac{h_i(z_3^{n_3},\ldots,z_l^{n_l})}
			{z_2^{n_2}}
			< z_1^{n_1}/z_2^{n_2} - a
			<  \frac{h_j(z_3^{n_3},\ldots,z_l^{n_l})}{z_2^{n_2}} \quad\text{for all $(i,j) \in X_1 \times X_2$}\\
			\vspace*{0.15cm} z_1^{n_1},z_2^{n_2} > z_3^{n_3} > \cdots > z_l^{n_l}\\
			\vspace*{0.15cm} h_i(z_3^{n_3},\ldots,z_l^{n_l}) > 0 \quad\text{for all $i \in I$}\\
			n_2-n_3 > N
		\end{cases}
	\end{equation}
	where $X_1, X_2, I$ are finite sets of indices, each $h_i$ is a $\rat$-linear form, $a \in \rat_{>0}$, and $N \in \nat$.
	Our algorithm for solving the system~\eqref{eq::overview-system} proceeds by first inductively solving the sub-system consisting of the inequalities $h_i(z_3^{n_3},\ldots,z_l^{n_l}) < h_j(z_3^{n_3},\ldots,z_l^{n_l})$ for all $(i,j) \in X$, $z_3^{n_3} > \cdots > z_l^{n_l}$, and $h_{i}(z_3^{n_3},\ldots,z_l^{n_l}) > 0$ for $i \in I$.
	If no such solution exists, then~\eqref{eq::overview-system} does not have a solution either.
	Otherwise, let $(m_3,\ldots,m_l)$ be a solution to the sub-system.
	In \Cref{sec::ineqs}, we use arguments from Diophantine approximation to prove that in the latter case the system~\eqref{eq::overview-system} does always have a solution, and show to construct such a solution from $(m_3,\ldots,m_l)$.

	In~\Cref{sec::hardness} we prove~\Cref{thm::main-hardness} and give a whole class of queries that become decidable assuming decidability of the existential fragment of $\mathcal{PA}(\alpha^x,\beta^x)$.
	These include occurrence of a given pattern in base-$\beta$ or base-$\alpha$ expansions of $\log_\alpha(\beta)$.
	Finally, in~\Cref{sec::undec} we show that for any multiplicatively independent $\alpha,\beta \in \nat_{>1}$, already for formulas with three alternating blocks of quantifiers (i.e.\ formulas with quantifier alternation depth 2) the membership problem in $\mathcal{PA}(\alpha^\nat,\beta^\nat)$ is undecidable.
	This result is included for the sake of describing the decidability landscape as completely as possible.
	We use the approach developed by Hieronymi and Schulz in \cite{hieronymi-schulz}, but reduce from the Halting Problem for 2-counter machines as opposed to the Halting Problem for Turing machines, which results in a simpler construction.
	Thus for any multiplicatively independent $\alpha,\beta$, the decidability question for $\mathcal{PA}(\alpha^\nat,\beta^\nat)$ remains open only for formulas containing exactly two alternating blocks of quantifiers.

	\section{From formulas to systems of inequalities}
	\label{sec::first-step}
	We now prove \Cref{thm::pres-to-system-reduction}.
	Our main tool is the fact that semilinear sets have quantifier-free representations constructed from linear inequalities and divisibility constraints.
	We note that our reduction from the decision problem for the existential fragment of $\mathcal{PA}(\alpha^\nat,\beta^\nat)$ to Problem~\ref{problem1} does not preserve $\alpha$ and~$\beta$.

	\begin{proof}[Proof of \Cref{thm::pres-to-system-reduction}]
		Suppose we are given $\alpha,\beta> 1$ and an existential formula $\exists \mathbf{z} \st \varphi(\mathbf{z})$ in the language $\Lcal_{\alpha,\beta}$,  where $\varphi$ is quantifier-free and $\mathbf{z}$ is a collection of variables.
		Before inspecting whether $\alpha, \beta$ are multiplicatively independent, we will apply a sequence of transformations to $\exists \mathbf{z} \st \varphi(\mathbf{z})$.
		For a term $t$ and $\gamma \in \{\alpha,\beta\}$, we can rewrite the formula $\lnot(t \in \gamma^{\nat})$ as
		\[
		t < 1 \lor \exists x\st x \in \gamma^\nat \land x < t < \underbrace{x + \cdots + x}_{\gamma \textrm{ times}}.
		\]
		Since $\lnot(t > 0)$ and $\lnot (t = 0)$ are equivalent to $t < 0 \lor t = 0$ and $t >0 \lor t <0$ respectively, we can construct a formula $\exists \mathbf{x} \st \widehat{\varphi}(\mathbf{x})$ equivalent to $\exists \mathbf{z} \st \varphi(\mathbf{z})$ in which the negation symbol does not occur.
		We can also rewrite $t \in \gamma^\nat$ as $y = t \land y \in \gamma^\nat$, where~$y$ is a fresh variable.
		Therefore, we can construct a formula
		\[
		\widetilde{\varphi}(\mathbf{y},\mathbf{x})\coloneqq \bigvee_{e \in E} \bigwedge_{j\in J_e}
		\mu_{j}(\mathbf{y},\mathbf{x})
		\]
		not containing the negation symbol,
		where $\mathbf{y}$ denotes a collection $y_1,\ldots,y_l$ of fresh variables, with the following properties.
		\begin{itemize}
			\item
			$\exists \mathbf{z} \in \intg^k
			\st  \varphi(\mathbf{z}) \Leftrightarrow \exists\mathbf{y} \in \intg^l, \mathbf{x} \in \intg^k \st \widetilde{\varphi}(\mathbf{y},\mathbf{x})$.
			\item
			For each $y_i$, there exists unique $\gamma_i \in \{\alpha,\beta\}$ such that $y_i \in \gamma_i^{\nat}$ is a sub-formula of $\widetilde{\varphi}$.
			\item
			Each $\mu_{j}(\mathbf{y},\mathbf{x})$ is an atomic formula either of the form $t(\mathbf{y},\mathbf{x}) \sim 0$ for a term $t(\mathbf{y}, \mathbf{x})$
			and ${\sim} \in \{>,=\}$, or of the form $y_i \in \gamma_i^{\nat}$ for some $i$.
		\end{itemize}
		Next, write each $\bigwedge_{j\in J_e}
		\mu_{j}(\mathbf{y},\mathbf{x})$ in the form
		\[
		\bigwedge_{j \in A_e} y_{\sigma(j)} \in \gamma_{\sigma(j)}^{\nat} \:\land\:
		 \bigwedge_{j \in B_e} t_{j}(\mathbf{y},\mathbf{x}) \sim_{j} 0
		\]
		where $\sigma(j) \in \{1,\ldots,l\}$ and ${\sim}_{j} \in \{>,=\}$ for all $j$.
		We can then write
		$ \exists\mathbf{y} \in \intg^l, \mathbf{x} \in \intg^k \st \widetilde{\varphi}(\mathbf{y},\mathbf{x})$ equivalently as
		\begin{equation}
			\label{eq::helper-2325}
			\bigvee_{e \in E}
			\biggl(
			\exists \mathbf{y} \in \intg^l\st
			\bigwedge_{j \in A_e} y_{\sigma(j)} \in \gamma_{\sigma(j)}^{\nat}
			\:\land\:
			\exists \mathbf{x} \in \intg^k \st \bigwedge_{j \in B_e} t_{j}(\mathbf{y},\mathbf{x}) \sim_{j} 0
			\biggr).
		\end{equation}
		For $e \in E$, let $S_e$ be the set of all $\mathbf{y} \in \intg^l$ such that $ \exists \mathbf{x} \in \intg^k \st \bigwedge_{j \in B_e} t_{j}(\mathbf{y},\mathbf{x}) \sim_{j} 0$ holds.
		Observe that each $S_e$ is semilinear.
		Setting $z_i = \gamma_i$ and $y_i = z_i^{n_i}$ for $1 \le i \le l$, we can rewrite \eqref{eq::helper-2325} as
		\[
		\bigvee_{e \in E}
		\exists n_1,\ldots, n_l \in\nat\st (z_1^{n_1},\ldots,z_l^{n_l}) \in S_e
		\]
		which is equivalent to
		\[
		\exists n_1,\ldots, n_l \in \nat \st (z_1^{n_1},\ldots,z_l^{n_l}) \in S
		\]
		for semilinear $S = \bigcup_{e \in E} S_e$.

		Recall from \Cref{sec::prelims} that each semilinear set has a representation in the form~\eqref{eq::pres-qe-first}.
		For $x, y, r \ge 0$  and $\lambda, D \ge 1$, note that $x + y \equiv r \bmod D$ is equivalent to
		\[
		\bigvee_{
		\substack{
		0 \le r_1,r_2 < D\\
		r_1 + r_2 \equiv r \bmod D}
		}
		x \equiv r_1 \bmod D \:\land\: y \equiv r_2 \bmod D
		\]
		and $x \equiv r \bmod D$ is equivalent to $\bigvee_{k=0}^{\lambda-1} x \equiv r + kD \bmod \lambda D$.
		Hence we can construct $D \ge 1$ and a representation of $S$ of the form
		\begin{equation}
			\label{eq::helper-1601}
			\bigvee_{p \in P}
			\biggl(
			\,
			\bigwedge_{i=1}^l x_{i} \equiv r_{i,p} \bmod D \:\:\land\:\:  \bigwedge_{s \in S_p} h_{s}(x_1,\ldots,x_d) \sim_{s} b_s
			\biggr)
		\end{equation}
		where each $r_{i,p} \ge 0$, $h_s$ is a $\intg$-linear form, $b_s \in \intg$, and ${\sim}_s \in \{>, =\}$.
		Write $\widetilde{S}_p$ for the set defined by $p \in P$ in~\eqref{eq::helper-1601}, so that $S = \bigcup_{p\in P} \widetilde{S}_p$.
		It suffices to reduce deciding $\exists n_1,\ldots,n_l \in \nat \st (z_1^{n_1},\ldots,z_l^{n_l}) \in \widetilde{S}_p$ to either Problem~\ref{problem1} or deciding $\mathcal{PA}(\gamma)$ for some $\gamma$, depending on whether $\alpha,\beta$ are multiplicatively independent.
		To do this, first observe that for $\gamma\in \{\alpha,\beta\}$, the sequence $\seq{\gamma^n \bmod D}$ is ultimately periodic, with the additional property that if $x$ occurs at least twice in the sequence, then it occurs infinitely often.
		Next, construct $D_\alpha, D_\beta$ such that
		\begin{itemize}
			\item[(i)] $\seq{\alpha^n \bmod D}$ is ultimately periodic with period $D_\alpha$,
			\item[(ii)] $\seq{\beta^n \bmod D}$ is ultimately periodic with period $D_\beta$, and
			\item[(iii)] if $\alpha,\beta$ are multiplicatively dependent, then $\alpha^{D_\alpha} = \beta^{D_\beta}$.
		\end{itemize}
		Such $D_\alpha, D_\beta$ can always be constructed because if a sequence is ultimately periodic with period~$k$, then it is ultimately periodic with period $km$ for every $m > 0$.
		We have that for all $1\le i \le l$ and $p \in P$, $z_i^{n_i} \equiv r_{i,p} \bmod D$ is either false for all $n_i$, true for exactly one value of $n_i$, or true on a union of arithmetic sequences with period $D_{z_i}$.
		Therefore, $\exists n_1,\ldots,n_l\st (z_1^{n_1},\ldots,z_l^{n_l}) \in \widetilde{S}_p$ can be equivalently expressed as a disjunction of formulas of the form
		\[
		\exists m_1,\ldots,m_l \in \nat \st \bigwedge_{s \in S_p}
		h_s \biggl(
		z_1^{t_{s,1}},\ldots,z_l^{t_{s,l}}
		\biggr)
		\sim_s b_s
		\]
		where for all $s,l$, $t_{s,l} = a$ or $t_{s,l} = a +m_i \cdot D_{z_i}$ for a constant $a \in \nat$.
		It remains to observe that $z_i^{a_i + m_i \cdot D_{z_i}} = z_i^{a_i} \left(z_i^{D_{z_i}}\right)^{m_i}$.
		Therefore, we have reduced deciding the truth value of $\exists \mathbf{x} \st \varphi(\mathbf{x})$ to solving systems of (in)equalities in
		\begin{itemize}
			\item powers of $\gamma \coloneqq \alpha^{D_\alpha}$ in case $\alpha,\beta$ are multiplicatively dependent, and
			\item powers of $\gamma_\alpha \coloneqq \alpha^{D_\alpha}$, $\gamma_\beta \coloneqq \beta^{D_\beta}$ otherwise.
		\end{itemize}
		Note that if $\alpha,\beta$ are multiplicatively independent, then $\gamma_\alpha,\gamma_\beta$ are also multiplicatively independent.
		This concludes the proof.
	\end{proof}

	\section{Solving Diophantine equations}
	\label{sec::equations}

	We now discuss solutions of systems of affine Diophantine equations in powers of $\alpha$ and $\beta$.
	This is the first step towards showing decidability of Problem~\ref{problem1}.
	Our goal in this section is to prove the following theorem.
	\equationsmain*

	First let us consider the easier case where $\alpha = \beta$.

	\begin{theorem}
		\label{thm::single-predicate-equations}
		Let $\alpha \in \nat_{>1}$, $l \ge 1$, $c_1,\ldots,c_l \in \intg_{\ne0}$, and $d \in \intg$.
		Further let $\Scal \subseteq \nat^l$ be the set of solutions of
		\[
		c_1\alpha^{n_1} + \cdots + c_l\alpha^{n_l} = d.
		\]
		Then $\Scal \in \mathfrak{A}$, and a representation of $\Scal$ in the form~\eqref{eq::aclass-form} can be effectively computed.
	\end{theorem}
	\begin{proof}
		By a case analysis on the ordering of $z_1^{n_1}, \ldots,z_l^{n_l}$, it suffices to show that for any permutation $\sigma \st \{1,\ldots,l\} \to \{1,\ldots,l\}$, the set $\widetilde{\Scal}$ of all $(n_1,\ldots,n_l) \in \nat^l$ such that
		\begin{equation*}
			\begin{cases}
				\vspace*{0.15cm}c_{1}\alpha^{n_{1}} + \cdots + c_{l}\alpha^{n_{l}} = d\\
				\vspace*{0.15cm} n_{\sigma(1)} \ge \cdots \ge n_{\sigma(l)}
			\end{cases}
		\end{equation*}
		belongs to $\mathfrak{A}$ and has an effectively computable representation.
		We prove this by induction on $l$.
		For $l  = 1$, the statement is immediate.
		Suppose $l \ge 2$.
		Because we can rename the variables, it suffices to consider $\sigma(j) = j$ for all $j$.
		Let $N$ be such that $\alpha^N > |d| + \sum_{i=1}^l |c_i|$. Then for every $(n_1,\ldots,n_l) \in \widetilde{\Scal}$ we have $0 \le n_1 - n_2 \le N$ and let $\widetilde{\Scal} = \bigcup_{k = 0}^N \widetilde{\Scal}_k$, where each $\widetilde{\Scal}_k$ is the set of all solutions of
		\begin{equation*}
			\begin{cases}
				\vspace*{0.15cm}c_{1}\alpha^{n_{1}} + \cdots + c_{l}\alpha^{n_{l}} = d\\
				\vspace*{0.15cm} n_2 \ge \cdots \ge n_l\\
				\vspace*{0.15cm} n_{1}  = n_{2} + k.
			\end{cases}
		\end{equation*}
		To construct a representation of $\widetilde{\Scal}_k$, we have to eliminate the variable $n_1$ using the last equation above.
		To do this, inductively solve the system
		\begin{equation*}
			\label{eq::helper-1801}
			\begin{cases}
				\vspace*{0.15cm} (c_{1}\alpha^k + c_2)\alpha^{n_{2}} +  c_{3}\alpha^{n_{3}} + \cdots + c_{l}\alpha^{n_{l}} = d\\
				\vspace*{0.15cm} n_{2} \ge \cdots \ge n_{l}.
			\end{cases}
		\end{equation*}
		and then add the condition $n_1 = n_2 + k$.
	\end{proof}

	Next, we show how to solve, under certain assumptions, equations involving powers of both $\alpha$ and $\beta$ by applying Baker's theorem on linear forms in an iterative fashion.
	These assumptions will be lifted later when proving \Cref{thm::equalities-main}.
	\begin{theorem}
		\label{thm::equalities-big-thm}
		Let $\alpha,\beta \in \nat_{>1}$ be multiplicatively independent, $l\ge2$, $z_1,\ldots,z_l\in \{\alpha,\beta\}$ with $z_1 = \alpha$ and $z_2=\beta$, $c_1,\ldots,c_l\in\intg_{\ne 0}$, and $d \in \intg$.
		Denote by $\Scal$ the set of all $(n_1,\ldots,n_l)\in\nat^l$ satisfying all of the following.
		\begin{enumerate}
			\item[(a)] $c_1z_1^{n_1}+\cdots+c_lz_l^{n_l}=d$;
			\item[(b)] $z_1^{n_1}, z_2^{n_2} \ge z_3^{n_3}\ge \cdots\ge z_l^{n_l}$;
			\item[(c)] For every non-empty proper subset $I$ of $\{1,\ldots,l\}$ it holds that $\sum_{i \in I}c_iz_i^{n_i} \ne 0$.
		\end{enumerate}
		Define $\mu(j)$ to be $1$ if $z_j = \alpha$ and $\mu(j) = 2$ if $z_j =\beta$.
		We have the following.
		\begin{enumerate}
			\item[(i)] We can compute $\xi_1,\xi_2 \in \mathbb{Q}$ such that $n_1 \ge \frac{\log(\beta)}{\log(\alpha)}n_2 - \xi_1$ and $n_2 \ge \frac{\log(\alpha)}{\log(\beta)}n_1 - \xi_2$.
			\item[(ii)] There exist effectively computable polynomials $p_1,\ldots,p_l \in \rat[x,y]$ such that
			\[
			n_{\mu(j)} - n_j < p_j(\log(1+n_1), \log(1+n_2))
			\] for all $(n_1,\ldots,n_l) \in \Scal$ and $3 \le j \le l$.
			\item[(iii)] The set $\Scal$ is finite and can be effectively computed.
		\end{enumerate}
	\end{theorem}
	\begin{proof}
		Observe that $n_j \le n_{\mu(j)}$ for all $j\ge 1$ by~(b).
		Let
		\[
		\Scal_1 = \big\{(n_1,\ldots,n_l) \in \Scal \st n_1 > n_2\big\}
		\] and $\Scal_2 = \Scal \setminus \Scal_1$.

		\emph{Proof of~(i).}
            Together~(a) and~(b) imply that, for all $(n_1,\ldots,n_l) \in \Scal$,
            \begin{equation*}
                \biggl(|d| + \sum_{i=2}^l|c_i|\biggl)z_2^{n_2}  \ge |c_1|z_1^{n_1}.
            \end{equation*}
            Taking logarithms and dividing by $\log(\beta) = \log(z_2)$ gives
            \begin{equation*}
                n_2 \ge \frac{\log(\alpha)}{\log(\beta)} n_1 - \frac{\log\left(|d| + \sum_{i=2}^l|c_i|\right)-\log|c_1|}{\log(\beta)}.
            \end{equation*}
            This allows us to find $\xi_2$.
            To compute $\xi_1$, observe that by~(a) and~(b),
            \[
            \biggl(|d| + |c_1| + \sum_{i=3}^l|c_i|\biggr)z_1^{n_1}  \ge |c_2|z_2^{n_2}
            \]
            and proceed similarly.

		\emph{Proof of~(ii).} By finite induction.
		Note that we can choose $p_1(x,y),p_2(x,y) = 1$.
		Suppose therefore $p_1,\ldots,p_j$ have already been computed for some $j \ge 2$.
		By swapping $z_1$ and $z_2$  if necessary, we can assume $z_{j+1} = z_1 =  \alpha$.
		(Observe that the roles $z_1$ and $z_2$ in the statement of our theorem are completely symmetrical.)
		For $(n_1,\ldots,n_l) \in \Scal$ define
		\begin{align*}
			X &\coloneqq-
			\sum_{
				\substack{1\le i \le j \\ z_i = \alpha}
			} c_i \alpha^{n_i-n_1},\\
			Y &\coloneqq \sum_{
				\substack{1\le i \le j \\ z_i = \beta}
			} c_i \beta^{n_i-n_2},\\
			\Lambda &\coloneqq \alpha^{-n_1}\beta^{n_2}X^{-1}Y-1
			=
			\biggl(
			-
			\sum_{
				\substack{1\le i \le j \\ z_i = \alpha}
			} c_i \alpha^{n_i}\biggr)^{-1} \cdot \sum_{\substack{1\le i \le j \\ z_i = \beta}} c_i \beta^{n_i} - 1.
		\end{align*}
		By~(c), $X$ is non-zero and hence $X^{-1}$ is well-defined, and similarly, $Y$ and $\Lambda$ are non-zero.
		Next, observe that~(a) can be written as
		\begin{equation}
			\label{eq::a-rewritten}
			\Lambda =
			\biggl(
			\,
			\sum_{i=j+1}^l c_i z_i^{n_i} - d
			\,
			\biggr)
			\cdot
			\biggl(
			\,
			\sum_{
				\substack{1\le i \le j \\ z_i = \alpha}
			} c_i \alpha^{n_i}
			\,
			\biggr)^{-1}.
		\end{equation}
		We will estimate the magnitude of terms on both sides of this equation, starting with the left-hand side.
		Recall the definition and the properties of the height function $h(\cdot)$ given in \Cref{sec::prelims}.
		We have that $h(X^{-1}Y) \le h(X) + h(Y)$ and
		\begin{align*}
			h(X) &\le  \log(j) + \sum_{
				\substack{1\le i \le j \\ z_i = \alpha}
			} \log |c_i| + (n_1 - n_i) \log |\alpha|,
			\\
			h(Y) &\le \log(j) + \sum_{
				\substack{1\le i \le j \\ z_i = \beta}
			} \log |c_i| + (n_2 - n_i) \log |\beta|.
		\end{align*}
		Therefore, using \Cref{thm::baker} we can compute $\kappa_1 \in \rat_{>0}$ such that
		\begin{equation}
			\label{eq::baker-applying}
			\log |\Lambda| > -\kappa_1 \cdot \big(1+\log(1+\max\, \{n_1,n_2\})\big) \cdot \max_{1\le i \le j} \{n_{\mu(i)}-n_i\}.
		\end{equation}
		Applying the induction hypothesis, there exists computable $q \in \rat[x,y]$ such that
		\[
		\log |\Lambda| > -\kappa_1\cdot q\big(\log(1+n_1),\log(1+n_2)\big).
		\]
		Next, consider the right-hand side of \eqref{eq::a-rewritten}.
		Let $a$ be the largest integer $1 \le i \le j$ such that $z_i = \alpha$.
		We have that
		\[
		\biggl\vert
		\,
		\sum_{i=j+1}^l c_i z_i^{n_i} - d
		\,
		\biggr\vert \le \kappa_2 \alpha^{n_{j+1}}
		\]
		for some computable $\kappa_2 \in \intg_{>0}$
		and, by~(c) and the  induction hypothesis,
		\[
		\biggl\vert
		\,
		\sum_{
			\substack{1\le i \le j \\ z_i = \alpha}
		} c_i \alpha^{n_i}
		\,
		\biggr\vert > \alpha^{n_a} >  \alpha^{n_1 - r(\log(1+n_1),\,\log(1+n_2))}
		\]
		for a polynomial $r \in \rat[x,y]$.
		Hence the magnitude of the right-hand side of \eqref{eq::a-rewritten} is bounded by $\kappa_2 \alpha^{r(\log(1+n_1),\,\log(1+n_2))-n_1+n_{j+1}}$, and a necessary condition for~\eqref{eq::a-rewritten} to hold is that
		\begin{equation*}
			-\kappa_1\cdot q\big(\log(1+n_1),\log(1+n_2)\big)
			< \log
			\bigl(
			\kappa_2 \cdot \alpha^{r(\log(1+n_1),\,\log(1+n_2))-n_1 + n_{j+1}}
			\bigr)
		\end{equation*}
		which is equivalent to
		\begin{equation}
			\label{eq::helper-2}
			n_1-n_{j+1} < \frac{\kappa_1\cdot q\big(\log(1+n_1),\log(1+n_2)\big)-\log(\kappa_2)}{\log(\alpha)} + r\big(\log(1+n_1),\log(1+n_2)\big).
		\end{equation}
		It remains to choose $p_{j+1} \in \rat[x,y]$ such that $p_{j+1}(\log(1+n_1),\log(1+n_2))$ is at least as large as the right-hand side of~\eqref{eq::helper-2}.

		\emph{Proof of~(iii).}
		Since $\alpha \ne \beta$ by the multiplicative independence assumption, without loss of generality we can assume that $\alpha < \beta$.
		(The roles of $\alpha = z_1$ and $\beta = z_2$ are symmetric and we can swap them if necessary.)
		Recall the definitions of $\Scal_1,\Scal_2$.
		Elements of $\Scal_2$ can be bounded using~(i),  as $\frac{\log(\beta)}{\log(\alpha)} > 1$ and $n_2 \ge n_1 \ge \frac{\log(\beta)}{\log(\alpha)}n_2 - \xi_1$ together yield a bound on $n_2$.
		It remains to bound $\Scal_1$.

		\emph{Case 1: Suppose $d \ne 0$.}
		As in the proof of (ii), let
		\begin{align*}
			X&=-\sum_{
				\substack{1\le i \le l \\ z_i = \alpha}
			} c_i \alpha^{n_i-n_1},\\
			Y&=\sum_{
				\substack{1\le i \le l \\ z_i = \beta}
			} c_i \beta^{n_i-n_2},\\
			\Lambda&= \alpha^{-n_1}\beta^{n_2}\cdot
			X^{-1}Y-1.
		\end{align*}
		Similarly to the proof of~(ii), $X$, $Y$, and $\Lambda$ are non-zero, and we rewrite (a) in the form
		\begin{equation}
			\label{eq::a-rewritten-2}
			\Lambda= -d \cdot \biggl(
			\,
			\sum_{
				\substack{1\le i \le l \\ z_i = \alpha}
			} c_i \alpha^{n_i}
			\,
			\biggr)^{-1}
		\end{equation}
		and bound the magnitude on both sides.
		Because $n_1 > n_2 \ge 0$ for all solutions in $\Scal_1$, application of \Cref{thm::baker} and~(ii) yields,
		\[
		\log |\Lambda| > -\kappa_2 p\big(\log(n_1)\big)
		\]
		where $\kappa_2 > 0$ and $p \in \rat[x]$ are computed effectively.
		It remains to compute an upper bound for the right-hand side.
		Let $a$ be the largest integer $1 \le i \le l$ such that $z_i = \alpha$.
		We have that
		\[
		\biggl\vert
		\,
		\sum_{
			\substack{1\le i \le l \\ z_i = \alpha}
		} c_i \alpha^{n_i}
		\,
		\biggr\vert > \alpha^{n_{a}} >  \alpha^{n_1 - f(\log(n_1))}
		\]
		where $f \in \rat[x]$.
		Hence a necessary condition for $(n_1,\ldots,n_l) \in \Scal_1$ is
		\[
		\kappa_2 \cdot p\big(\log(n_1)\big) >   \big(n_1-f(\log(n_1))\big)\cdot \log(\alpha) - \log|d|
		\]
		from which we can compute a bound on $n_1$.
		Once we bound $n_1$, a bound on the remaining variables can be computed in the same way as above.

		\emph{Case 2: Suppose $d = 0$.}
		We will need a lemma.
		\begin{lemma}
			\label{thm::prime-p}
			There exists a prime $p \in\nat$ such that $\nu_p(\beta) > 0$ and
			\[
			\frac{\log(\alpha)}{\log(\beta)} > \frac{\nu_p(\alpha)}{\nu_p(\beta)}.
			\]
		\end{lemma}
		\begin{proof}
			If there is a prime $p$ dividing $\beta$ that does not divide $\alpha$, then the statement is immediate.
			Suppose therefore that $\alpha,\beta$ have exactly the same prime divisors $p_1,\ldots,p_k$.
			We have
			\[
			\frac{\log(\alpha)}{\log(\beta)} = \frac{\nu_{p_1}(\alpha)\log(p_1)+\cdots+\nu_{p_k}(\alpha)\log(p_k)}{\nu_{p_1}(\beta)\log(p_1)+\cdots+\nu_{p_k}(\beta)\log(p_k)}.
			\]
			By multiplicative independence, $\log(\alpha)/\log(\beta) \notin \rat$ and hence $\nu_{p_i}(\alpha)/\nu_{p_i}(\beta) \ne \log(\alpha)/\log(\beta)$ for all $1 \le i \le k$.
			It follows that $\log(\alpha)/\log(\beta) > \nu_{p_i}(\alpha)/\nu_{p_i}(\beta)$ for some $i$.
		\end{proof}

		To bound the elements of $\Scal_1$, let $a = \max \{i \st z_i = \alpha\}$ and $b =  \max \{i \st z_i = \beta\}$.
		Further let
		\begin{align*}
			A &= \sum_{\substack{1\le i \le l \\ z_i = \alpha}}c_i\alpha^{n_i},\\
			B &= -\sum_{\substack{1\le i \le l \\ z_i = \beta}} c_i\beta^{n_i}.
		\end{align*}
		Let $p$ be a prime as in \Cref{thm::prime-p}.
		A necessary condition for~(a) to hold is that
		\[
		\nu_p(A) \ge \nu_p(B).
		\]
		As discussed in \Cref{sec::prelims}, we have
		\[
		\nu_p(B) \ge \min_{z_i = \beta} \big\{\nu_p(c_i \beta^{n_i})\big\} \ge \nu_p(\beta^{n_{b}}) = n_{b} \cdot \nu_p(\beta).
		\]
		Under the assumption $n_1 > n_2$, by~(ii), there exists effectively computable $q_1 \in \rat[x]$ such that $n_{b} > n_2- q_1(\log(n_1))$.
		Hence,
		\[
		\nu_p(B) > n_2\nu_p(\beta) - q_1(\log(n_1))\nu_p(\beta).
		\]
		Meanwhile,
		\begin{align*}
			\nu_p(A) &= \nu_p(\alpha^{n_{a}}) + \nu_p
			\biggl(
			\,
			\sum_{\substack{1\le i \le l \\ z_i = \alpha}}c_i\alpha^{n_i-n_{a}}
			\biggr)\\
			&\le
			n_1\nu_p(\alpha) + \log_p
			\,
			\biggl|
			\sum_{\substack{1\le i \le l \\ z_i = \alpha}}c_i\alpha^{n_i-n_{a}}
			\biggr|.
		\end{align*}
		Applying~(ii), we obtain that
		\[
		\nu_p(A) \le n_1\nu_p(\alpha) + q_2\big(\log(n_1)\big)
		\]
		for an effectively computable $q_2 \in \rat[x]$.
		Thus, a necessary condition for~(a) to hold is that
		\[
		n_2\nu_p(\beta) - q_1\big(\log(n_1)\big)\nu_p(\beta) \le n_1\nu_p(\alpha) + q_2\big(\log(n_1)\big)
		\]
		which is equivalent to
		\[
		n_2-n_1\frac {\nu_p(\alpha)} {\nu_p(\beta)} \le \frac{ q_1\big(\log(n_1)\big)\nu_p(\beta) + q_2\big(\log(n_1)\big)}{\nu_p(\beta)}.
		\]
            Applying~(i), we obtain that
            \[
		-\xi_2 + n_1\left(\frac{\log(\alpha)}{\log(\beta)} -\frac {\nu_p(\alpha)} {\nu_p(\beta)}\right) \le \frac{ q_1\big(\log(n_1)\big)\nu_p(\beta) + q_2\big(\log(n_1)\big)}{\nu_p(\beta)}.
		\]
            By construction of~$p$, the left-hand side of the inequality above grows linearly in $n_1$ while the right-hand side grows poly-logarithmically.
            Hence we can compute a bound on $n_1$, from which bounds on every $n_i$ can be derived.
	\end{proof}

	We can now finalise the proof of the main result of this section.
	\begin{proof}[Proof of \Cref{thm::equalities-main}]
		Since the class $\mathfrak{A}$ is closed under intersections, it suffices to consider the case where $s = 1$, i.e.\ the case of a single equation of the form
		\begin{equation}
			\label{eq::helper-2237}
			c_{1} z_1^{n_1} + \cdots + c_{l}z_l^{n_l} = d.
		\end{equation}
		Denote the set of solutions by $\Scal$.
		The proof is by induction on $l$.
		The statement is immediate for $l = 1$.
		Suppose $l \ge 2$.
		Let $N$ be such that $\alpha^N, \beta^N > |d| + \sum_{i=1}^l |c_i|$.
		Define
		\[
		\Scal_1 = \big\{(n_1,\ldots,n_l) \in \nat^l \mid \exists a,b \st a \ne b \textrm{ and }z_{a} = z_b \textrm{ and }  0 \le n_a-n_b \le N \big\},
		\]
		and $\Scal_2 = \Scal \setminus \Scal_1$.
		Intuitively, for every solution in $\Scal_2$, the two dominant terms among $z_1^{n_1},\ldots,z_l^{n_l}$ must have different bases (as otherwise they would be too far apart in magnitude), which will allow us to apply \Cref{thm::equalities-big-thm}.
		Further let $\Mcal$ be the set of all non-empty proper subsets of $\{1,\ldots,l\}$, and $\Scal_\mu$ for $\mu \in \Mcal$ be the set of all $(n_1,\ldots,n_l) \in \Scal$ such that
		\[
		\sum_{i \in \mu} c_i z_i^{n_i} = 0.
		\]
		Finally, let $\widetilde{\Scal}$ be the set of all $(n_1,\ldots,n_l) \in \nat^l$ such that for all $\mu \in \Mcal$,
		\[
		\sum_{i \in \mu} c_i z_i^{n_i} \ne 0.
		\]
		That is, $\widetilde{\Scal}$ is the set of all solutions of \eqref{eq::helper-2237} where no proper sub-sum vanishes.
		We will express $\Scal$ in the form
		\[
		\Scal = \Scal_1 \cup \left(\Scal_2 \cap \widetilde{\Scal}\right) \cup  \bigcup_{\mu \in \Mcal} \Scal_\mu.
		\]
		Since each $\Scal_\mu$ is exactly the set of solutions to
		\[
		\sum_{i \notin \mu} c_iz_i^{n_i} = d,
		\]
		in which fewer variables than $l$ appear,
		we can apply the induction hypothesis.
		To compute a representation of $\Scal_1$, we will compute, for every  $0 \le k \le N$ and distinct $1 \le a, b \le l$ with $z_\alpha = z_\beta$, a representation of the set of all $(n_1,\ldots,n_l) \in \Scal$ satisfying $n_a = n_b + k$.
		To do this, we just have to eliminate the variable $n_a$.
		That is, inductively compute a representation of the set of solutions of
		\[
		\big(c_a z_a^k + c_b\big)z_b^{n_b} + \sum_{i \ne a,b} c_i z_i^{n_i} =d
		\]
		and add the condition $n_a = n_b + k$.

		It remains to describe the structure of~$\Scal_2\cap\widetilde{\Scal}$.
		Let $P$ be the set of all permutations of $\{1,\ldots,l\}$.
		For $p \coloneqq (p_1,\ldots,p_l)\in P$, define $\widetilde{\Scal}_p$ as the set of all $(n_1,\ldots,n_l) \in \Scal_2\cap\widetilde{\Scal}$ that satisfy
		\[
		z_{p_1}^{n_{p_1}}, z_{p_2}^{n_{p_2}} \ge z_{p_3}^{n_{p_3}} \ge  \cdots \ge z_{p_l}^{n_{p_l}}.
		\]
		For particular $p = (p_1,\ldots,p_l)$, if $z_{p_1} \ne z_{p_2}$, then we can invoke invoke \Cref{thm::equalities-big-thm} to construct a representation of~$\widetilde{\Scal}_p$
		On the other hand, if $z_{p_1} = z_{p_2}$, then by the construction of~$N$, either $n_{p_1} > n_{p_2} + N$ and hence $c_{p_1}z_{p_1}^{n_{p_1}}$ dominates the other summands, or  $n_{p_2} > n_{p_1} + N$ and the same argument applies.
		Hence in this case $\widetilde{\Scal}_p$ must be empty.
	\end{proof}

	\section{Handling inequalities}
	\label{sec::ineqs}
	In this section we prove decidability of Problem~\ref{problem1}.
	As discussed in \Cref{sec::overview}, this, in conjunction  with \Cref{thm::pres-to-system-reduction}, completes the proof of our main decidability result (\Cref{thm::main-decidability}).
	The following lemma is one of our main technical tools.
	In particular, it says that if $A\mathbf{z} > \zerovec$ has a solution, then it has infinitely many solutions.

	\begin{lemma}[Pumping Lemma]
		\label{thm::pumping}
		Suppose we are given
		\begin{itemize}
			\item[(a)] $\rat$-linear forms $h_1,\ldots,h_r$ in $l \ge 1$ variables,
			\item[(b)] multiplicatively independent $\alpha,\beta \in \nat_{>1}$,
			\item[(c)] $z_1,\ldots,z_l$ satisfying $z_i \in \{\alpha,\beta\}$ for all $i$ and $z_1 = \beta$,
			\item[(d)] $m_1,\ldots,m_l \in \nat$, and
			\item[(e)] $\varepsilon \in \rat_{>0}$.
		\end{itemize}
		Write $J = \{j \st h_j(z_1^{m_1},\ldots,z_l^{m_l}) > 0\}$.
		We can compute 
		$\mu, \delta \in \rat_{>0}$ with the following property.
		Suppose $n_1 > m_1$ is such that there exists $k \in \nat$ for which $|\alpha^k/\beta^{n_1} - \mu| < \delta$.
		Then there exist $n_2,\ldots,n_l$ such that for all $1 \le j \le r$,
		\begin{itemize}
			\item[(i)] if $j \in J$, then $h_j(z_1^{n_1},\ldots,z_l^{n_l}) > 0$, and
			\item[(ii)] $\left\vert \frac{h_j\big(z_1^{n_1},\ldots,z_l^{n_l}\big)}{z_1^{n_1}} - \frac{h_j\big(z_1^{m_1},\ldots, z_l^{m_l}\big)}{z_1^{m_1}} \right\vert < \varepsilon$.
		\end{itemize}
		In particular, there exist infinitely many $n_1$ that can be extended to $(n_1,\ldots,n_l)$ satisfying (i) and~(ii) for all $1 \le j \le r$.
	\end{lemma}
	\begin{proof}
		By re-ordering $z_2,\ldots,z_l$, we can without loss of generality assume that $z_1,\ldots,z_b = \beta$ and $z_{b+1},\ldots,z_l = \alpha$ for some $1 \le b \le l$.
		For $1 \le j \le r$, write
		\[
		h_j(x_1,\ldots,x_l) =
		 t_j(x_1,\ldots,x_b) + s_j(x_{b+1},\ldots,x_l)
		\]
		where $s_j,t_j$ are $\rat$-linear forms.
		Let $\nu \in \rat_{>0}$ be such that
		\begin{itemize}
			\item[(A)] $t_j(\beta^{m_1}, \ldots, \beta^{m_b}) + c\cdot s_j(\alpha^{m_{b+1}},\ldots,\alpha^{m_l}) > 0$ for all $c \in (1-\nu, 1 + \nu)$ and $j \in J$, and
			\item[(B)] $\nu \cdot
			\left|
			\frac{s_j\big(\alpha^{m_{b+1}}, \ldots, \alpha^{m_{l}}\big)}{\beta^{m_1}}
			\right|
			< \varepsilon$ for all $1 \le j \le r$.
		\end{itemize}
		Choose $\mu = 1 / \beta^{m_1}$ and $\delta \in \rat$ such that $0 < \delta \beta^{m_1} < \nu$.
		It remains to argue the correctness of our choice of $\mu, \delta$.
		
		Suppose $n_1 > m_1$ is such that $|\alpha^k/\beta^{n_1} - \mu| < \delta$ for some $k \in \nat$.
		Write $m_\beta = n_1 - m_1$ and $m_\alpha = k$.
		We have that
		\[
		\biggl |
		\frac{\alpha^{m_\alpha}}{\beta^{m_\beta}} - 1
		\biggr |
		= {\beta^{m_1}}
		\biggl |
		\frac{\alpha^{m_\alpha}}{\beta^{n_1}} - \mu
		\biggr | < \nu.
		\]
		For $2 \le i \le l$, define $n_i = m_i + m_\beta$ if $z_i = \beta$ and  $n_i = m_i + m_\alpha$ if $z_i = \alpha$.
		Then, for all $j \in J$,
		\begin{align*}
			\frac{1}{\beta^{m_\beta}}h_j(z_1^{n_1},\ldots,z_l^{m_l}) &=
			\frac{\alpha^{m_\alpha}}{\beta^{m_\beta}} s_j(\alpha^{m_{b+1}},\ldots,\alpha^{m_l}) +  t_j(\beta^{m_1}, \ldots, \beta^{m_{b}}) \\
			&> 0
		\end{align*}
		where the inequality follows from~(A).
		This proves~(i).
		To prove~(ii), first observe that for $1 \le i \le l$, $z_i^{n_i}/z_1^{n_1} = c_i z_i^{m_i}/z_1^{m_1}$ where $c_i = 1$ if  $z_i = \beta$ and $c_i = \alpha^{m_\alpha}/\beta^{m_\beta}$ of $z_i = \alpha$.
		Hence
		\[
		\frac{t_j(\beta^{n_1}, \ldots, \beta^{n_{b}})}{z_1^{n_1}}
		=
		\frac{t_j(\beta^{m_1}, \ldots, \beta^{m_{b}})}{z_1^{m_1}}
		\]
		for all $1 \le j \le r$.
		Therefore, for all $j$,
		\begin{align*}
			\frac{h_j(z_1^{n_1},\ldots,z_l^{n_l})}{z_1^{n_1}} - \frac{h_j(z_1^{m_1},\ldots,z_l^{m_l})}{z_1^{m_1}}
			&=
			\frac{s_j(\alpha^{n_{b+1}}, \ldots, \alpha^{n_{l}})}{\beta^{n_1}}
			-
			\frac{s_j(\alpha^{m_{b+1}}, \ldots, \alpha^{m_{l}})}{\beta^{m_1}}
			\\
			&=\frac{s_j(\alpha^{m_{b+1}+m_\alpha}, \ldots, \alpha^{m_l+m_\alpha})}{\beta^{m_1+m_\beta}}
			-
			\frac{s_j(\alpha^{m_{b+1}}, \ldots, \alpha^{m_{l}})}{\beta^{m_1}}
			\\
			&=\frac{s_j(\alpha^{m_{b+1}}, \ldots, \alpha^{m_{l}})}{\beta^{m_1}}
			\biggl(
			\frac{\alpha^{m_\alpha}}{\beta^{m_\beta}} - 1
			\biggr).
		\end{align*}
		It remains to invoke~(B).
		Finally, that there exist infinitely many $n_1$ that can be extended to $(n_1,\ldots,n_l)$ satisfying (i) and~(ii) for all $j$ follows from~\Cref{thm::kronecker}.
	\end{proof}
	\begin{corollary}
		\label{thm::pumping-corollary}
		Let $\alpha,\beta \in \nat_{>1}$ be multiplicatively independent, $z_1,\ldots,z_l \in \{\alpha,\beta\}$, $A \in \intg^{r \times l}$, and $\mathbf{b} \ge 0$.
		There exists $\mathbf{z} = (z_1^{m_1},\ldots,z_l^{m_l})$ with $m_1,\ldots,m_l \ge 0$ satisfying $A\mathbf{z}>\zerovec$ if and only if there exists $\mathbf{\widetilde{z}} = (z_1^{n_1},\ldots,z_l^{n_l})$ with $n_1,\ldots,n_l \ge 0$ satisfying $A \mathbf{\widetilde{z}} > \mathbf{b}$.
	\end{corollary}
	\begin{proof}
        The ``if'' direction is trivial as one can take $\mathbf{z} := \widetilde{\mathbf{z}}$.
        Thus, we focus on the other direction.
		Let $\mathbf{z} = (z_1^{m_1},\ldots,z_l^{m_l})$ be as above.
		For $1 \le j \le r$, define the form
		\[
		h_j(x_1,\ldots,x_l) = e_j^\top A \cdot (x_1,\ldots,x_l).
		\]
		Let $\varepsilon \in \rat_{>0}$ be such that $h_j(z_1^{m_1},\ldots,z_l^{m_l})/z_1^{m_1} > 2\varepsilon$ for all $j$.
		Invoke \Cref{thm::pumping} with the forms $h_1,\ldots,h_r$ and the values $m_1,\ldots,m_l, \varepsilon$.
		We obtain that there exist infinitely many $(\widetilde{m}_1,\ldots,\widetilde{m}_l)$ (where $m_1$ can be taken to be arbitrarily large) such that $h_j(z_1^{\widetilde{m}_1},\ldots,z_l^{\widetilde{m}_l}) > 0$ and
		\begin{equation}
			\label{eq::helper-1543}
			\left|\frac{h_j\big(z_1^{\widetilde{m}_1},\ldots,z_l^{\widetilde{m}_l}\big)}{z_1^{\widetilde{m}_1}} - \frac{h_j\big(z_1^{m_1},\ldots,z_l^{m_l}\big)}{z_1^{m_1}}\right| < \varepsilon
		\end{equation}
		for all $j$.
		Since $h_j(z_1^{m_1},\ldots,z_l^{m_l})/z_1^{m_1} > 2\varepsilon$,
		inequality~\eqref{eq::helper-1543} implies that $h_j\big(z_1^{\widetilde{m}_1},\ldots,z_l^{\widetilde{m}_l}\big) > \varepsilon z_1^{\widetilde{m}_1}$.
		It remains to choose $(\widetilde{m}_1,\ldots,\widetilde{m}_l)$ with $z_1^{\widetilde{m}_1}$ sufficiently large.
	\end{proof}

	The following is a useful lemma showing how to eliminate a variable $n_a$ if we can bound the gap between $n_a$ and some other (suitable) variable $n_b$.

	\begin{lemma}
		\label{thm::eliminating-one-var}
		Let $\alpha,\beta \in\nat_{>1}$, $z_1,\ldots,z_l \in \{\alpha,\beta\}$ for $l \ge 2$, and $1\le a,b \le l$ be distinct with $z_a = z_b$.
		Suppose we are given the system
		\begin{equation}
			\label{eq::helper-1850}
			\begin{cases}
				\vspace*{0.15cm}A\mathbf{z} > \zerovec\\
				\vspace*{0.15cm} N_1 \le n_a - n_b \le N_2
			\end{cases}
		\end{equation}
		where $A \in \intg^{r \times l}$ for $r \ge1$, $\mathbf{z} = (z_1^{n_1}, \ldots, z_l^{n_l})$ and $N_1,N_2 \ge 0$.
		Then we can construct matrices $\widetilde{A}_k \in \intg^{r\times (l-1)}$ for $N_1\le k \le N_2$ and $y_1,\ldots,y_{l-1} \in \{\alpha,\beta\}$ with the following property.
		There exists $\mathbf{y} = \big(y_1^{n_1},\ldots,y_{l-1}^{n_{l-1}}\big)$ satisfying   $n_1,\ldots,n_{l-1}\ge 0$ and
		\[
		\bigvee_{ k = N_1}^{N_2} \widetilde{A}_k \mathbf{y} > \zerovec
		\] if and only if the system~\eqref{eq::helper-1850} has a solution.
	\end{lemma}
	\begin{proof}
		Choose $(y_1,\ldots,y_{l-1})$ to be any ordering of $\{z_1,\ldots,z_l\} \setminus \{z_a\}$.
		It suffices to construct $\widetilde{A}_k$ for $N_1 \le k \le N_2$ such that $\widetilde{A}_k \cdot \mathbf{y} > \zerovec$ has a solution if and only if
		\begin{equation}
			\label{eq::helper-2139}
			\begin{cases}
				\vspace*{0.15cm}A\mathbf{z} > \zerovec\\
				\vspace*{0.15cm}  n_a  = n_b + k
			\end{cases}
		\end{equation}
		has a solution.
		The system~\eqref{eq::helper-2139} has a solution if and only if there exist $n_1,\ldots,n_{a-1},n_{a+1},\ldots,n_l$ such that
		\begin{equation}
			\label{eq::helper-2153}
			\big(A_{j,a}z_b^k + A_{j,b}\big)z_b^{n_b} + \sum^l_{
				\substack{
					i = 1\\
					i \ne a,b
				}
			} A_{j,i}z_i^{n_i} > 0
		\end{equation}
		for all $1 \le j \le r$.
		Thus we have eliminated the variable $n_a$, and can construct $\widetilde{A}_k$ by writing~\eqref{eq::helper-2153} for $1 \le j \le r$ in the matrix form.
	\end{proof}

	By \Cref{thm::pumping-corollary}, to solve the inequality $A \mathbf{z} > \mathbf{b}$ for $\mathbf{b} \ge 0$ it suffices to solve $A\mathbf{z} > \zerovec$.
	Next we show how to do the latter.

	\begin{restatable}{theorem}{ineqsmain}
		\label{thm::ineq-pure}
		Suppose we are given multiplicatively independent $\alpha,\beta \in \nat_{>1}$, $z_1,\ldots,z_l \in \{\alpha,\beta\}$ for some $l \ge 1$, and $A \in \intg^{r \times l}$ with $r > 0$.
		It is decidable whether there exist $n_1,\ldots,n_l \in \nat$ such that $A\mathbf{z} > \zerovec$, where $\mathbf{z} = (z_1^{n_1},\ldots,z_l^{n_l})$.
	\end{restatable}
	\begin{proof}
		The proof is by induction on $l$.
		For $l = 1$, the statement is immediate.
		Suppose $l = 2$.
		Then $A\mathbf{z} > \zerovec$ is equivalent to
		$z_1^{n_1}/z_2^{n_2} \in (c,d)$ for some $c,d\in \rat \cup \{+\infty\}$.
		If $z_1 = z_2$, then a solution exists if and only $z_1^k \in (c,d)$ for some $k\in \intg$, which is trivial to determine.
		If $z_1 \ne z_2$, then applying \Cref{thm::kronecker}, a solution exists if and only if $d > 0$ and $(c,d)$ is non-empty.

		Suppose $l > 2$.
		If we additionally assume that $z_a^{n_a} = z_b^{n_b}$ for some $a \ne b$, then we can eliminate at least one variable and solve the resulting system inductively, as follows.
		If $z_a = z_b$, then $n_a = n_b$, and we can invoke \Cref{thm::eliminating-one-var} with $N_1 = N_2 = 0$.
		If $z_a \ne z_b$, then by multiplicative independence, $n_a = n_b = 0$, and we can eliminate two variables.
		Hence we have reduced our problem to solving $l(l-1)/2$ systems in at most $l-1$ variables (which can be solved inductively), and the system
		\begin{equation}
			\label{eq::helper-0925}
			\begin{cases}
				\vspace*{0.15cm}A\mathbf{z} > \zerovec\\
				\vspace*{0.15cm} z_a^{n_a} \ne z_b^{n_b} \textrm{ for all $a\ne b$.}
			\end{cases}
		\end{equation}
		Next, by a case analysis on the largest two terms among $z_1^{n_1},\ldots,z_l^{n_l}$ and the order of the remaining terms, we reduce solving~\eqref{eq::helper-0925} to solving systems of the form
		\begin{equation*}
			\begin{cases}
				\vspace*{0.15cm}A\mathbf{z} > \zerovec\\
				\vspace*{0.15cm} z_{\sigma(1)}^{n_{\sigma(1)}}, z_{\sigma(2)}^{n_{\sigma(2)}} >  z_{\sigma(3)}^{n_{\sigma(3)}} >  \cdots > z_{\sigma(l)}^{n_{\sigma(l)}} \\
				\vspace*{0.15cm} z_{\sigma(1)}^{n_{\sigma(1)}} \ne z_{\sigma(2)}^{n_{\sigma(2)}}
			\end{cases}
		\end{equation*}
		where $\sigma$ is a permutation of $\{1,\ldots,l\}$.
		By renaming variables, we can rewrite the above as
		\begin{equation}
			\label{eq::helper-2320}
			\begin{cases}
				\vspace*{0.15cm}\widetilde{A} \mathbf{z} > 0\\
				\vspace*{0.15cm}z_1^{n_1},z_2^{n_2} >z_3^{n_3} > \cdots > z_l^{n_l}\\
				\vspace*{0.15cm}z_1^{n_1} \ne z_2^{n_2}.
			\end{cases}
		\end{equation}
		We will show how to solve such systems.

		Suppose $z_1 = z_2$.
		In this case we consider the two possibilities $n_1 > n_2$ and $n_1 < n_2$.
		We will only show how to solve the system
		\begin{equation}
			\label{eq::helper-1930}
		\begin{cases}
			\vspace*{0.15cm}\widetilde{A} \mathbf{z} > 0\\
			\vspace*{0.15cm}z_1^{n_1} > z_2^{n_2} >z_3^{n_3} > \cdots > z_l^{n_l}
		\end{cases}
		\end{equation}
		as the same argument applies to the case of $n_1 < n_2$.
		If $A_{j,1} < 0$ for some $j$, then we can compute~$N$ such that $1 \le n_1 - n_2 \le N$ in every solution of~\eqref{eq::helper-1930}.
		We can then eliminate the variable $n_1$ using \Cref{thm::eliminating-one-var}, and solve the resulting system in $l-1$ variables inductively.
		Now suppose $A_{j,1} \ge 0$ for all $j$.
		Let $K = \{k \st A_{k,1} = 0\}$ and $h_k(x_2,\ldots,x_l) = \sum_{i=2}^l A_{k,i} \cdot x$ for $k \in K$.
		Inductively solve the system consisting of the inequalities $z_2^{n_2} >\cdots >z_l^{n_l}$ and $h_k(z_2^{n_2},\ldots,z_l^{n_l}) > 0$ for $k \in K$.
		If there is no solution, then~\eqref{eq::helper-1930} does not have solution either.
		Otherwise, a solution to~\eqref{eq::helper-1930} can be constructed from the solution to the sub-system by choosing $n_1$ to be sufficiently large.

		Suppose $z_1  \ne z_2$; this is the more difficult case.
		By multiplicative independence, $z_1^{n_1} \ne z_2^{n_2}$ if and only if at least one of $n_1,n_2$ is non-zero.
		This is automatically satisfied, as $l > 2$ and $z_1^{n_1}, z_2^{n_2} > z_3^{n_3}$.
		By exchanging $z_1$ and $z_2$ if necessary, we can assume that $z_1 \ne z_2 = z_3$.
		Note that this implies $n_2 > n_3$.
		Further assume, without loss of generality, that $z_1 = \alpha$ and $z_2 = \beta$.

		By multiplying inequalities with different rational constants if necessary, write the system~\eqref{eq::helper-2320} in the form

		\begin{equation}
			\label{eq::helper-2318}
			\begin{cases}
				\vspace*{0.15cm} z_1^{n_1} > p_i(z_2^{n_2},\ldots,z_l^{n_l})  \quad \textrm{ for } i \in I_-\\
				\vspace*{0.15cm} z_1^{n_1} < p_i(z_2^{n_2},\ldots,z_l^{n_l}) \quad \textrm{ for } i \in I_+\\
				\vspace*{0.15cm} p_i(z_2^{n_2}, \ldots, z_l^{n_l}) > 0 \quad\quad \textrm{ for } i \in J\\
				\vspace*{0.15cm} z_1^{n_1},z_2^{n_2} > z_3^{n_3} > \cdots > z_l^{n_l}
			\end{cases}
		\end{equation}
		where $I_-, I_+, J$ are disjoint finite sets and each $p_i$ is a $\rat$-linear form.
		We can assume that $I_-$ is non-empty by adding the identically zero $\rat$-linear form over $l-1$ variables if necessary.
		Suppose $I_+$ is empty.
		Then inductively solve the sub-system
		\begin{equation}
			\begin{cases}
				\vspace*{0.15cm} p_i(z_2^{n_2}, \ldots, z_l^{n_l}) > 0 \quad\quad \textrm{ for } i \in J\\
				\vspace*{0.15cm} z_2^{n_2} > z_3^{n_3} > \cdots > z_l^{n_l}.
			\end{cases}
		\end{equation}
		If a solution exists, then a solution to \eqref{eq::helper-2318} can be constructed by choosing $n_1$ to be sufficiently large.
		Therefore, we can suppose both $I_-$ and $I_+$ are non-empty.

        Let $a_-$ be the largest coefficient of $z_2^{n_2}$ of any linear form $p_i(z_2^{n_2},\ldots,z_l^{n_l})$ with $i \in I_-$ and $\widetilde{I}_- \subseteq I_-$ all indices $i$ such that the coefficient of $z_2^{n_2}$ in $p_i(z_2^{n_2},\ldots,z_l^{n_l})$  equals $a_-$.
        Then one can effectively compute a number $N_-$ such that when $n_2 - n_3 > N_-$, $p_i(z_2^{n_2},\ldots,z_l^{n_l}) \le p_j(z_2^{n_2},\ldots,z_l^{n_l})$ for all $i \in I_-\setminus \widetilde{I}_-$ and $j \in \widetilde{I}_-$.
        Further, for $i \in \widetilde{I}_-$, write $p_i(z_2^{n_2},\ldots,z_l^{n_l}) = a_-z_2^{n_2} + h_i(z_3^{n_3},\ldots,z_l^{n_l})$.

        Similarly, let $a_+$ be the smallest coefficient of $z_2^{n_2}$ of any linear form $p_i(z_2^{n_2},\ldots,z_l^{n_l})$ with $i \in I_+$ and $\widetilde{I}_+ \subseteq I_+$ all indices $i$ such that the coefficient of $z_2^{n_2}$ in $p_i(z_2^{n_2},\ldots,z_l^{n_l})$  equals $a_+$.
        Then one can effectively compute a number $N_+$ such that $p_i(z_2^{n_2},\ldots,z_l^{n_l}) \le p_j(z_2^{n_2},\ldots,z_l^{n_l})$ for all $i \in I_+\setminus \widetilde{I}_+$ and $j \in \widetilde{I}_+$ when $n_2 - n_3 > N_+$.
        Further, for $i \in \widetilde{I}_+$, write $p_i(z_2^{n_2},\ldots,z_l^{n_l}) = a_+z_2^{n_2} + h_i(z_3^{n_3},\ldots,z_l^{n_l})$.

        For $i \in J$, let $c_i$ be the coefficient $z_2^{n_2}$ in $p_i(z_2^{n_2},\ldots,z_l^{n_l})$. Then let $\widetilde{J}$ be the subset of $J$ such that $c_i$ is zero and for $i \in \widetilde{J}$, write $h_i(z_3^{n_3},\ldots,z_l^{n_l}) = p_i(z_2^{n_2},\ldots,z_l^{n_l})$.  
        For some computably large~$N$, the following holds: If $i \in J \setminus \widetilde{J}$, then when $n_2 - n_3 > N$, the sign of $p_i(z_2^{n_2},\ldots,z_l^{n_l})$ equals the sign of $c_i$.
        Hence, if any $c_i$ is negative, we can reject the input when $n_2-n_3 > N$ and all inequalities where $c_i > 0$ are trivially satisfied.
        By taking $N$ large enough, we can assume that $N \ge N_-, N_+$.

        If we add $0 \le n_2 - n_3 \le N$ to~\eqref{eq::helper-2318}, we can solve the resulting system by eliminating $n_2$ using \Cref{thm::eliminating-one-var}.
        Meanwhile, if $n_2 - n_3 > N$, we have reduced ~\eqref{eq::helper-2318} (and hence our original decision problem) to solving systems of the following form:
		\begin{numcases}{}
			a_- + \frac{h_i(z_3^{n_3},\ldots,z_l^{n_l})}
			{z_2^{n_2}}
			< \frac{z_1^{n_1}}{z_2^{n_2}}
			< a_+ + \frac{h_j(z_3^{n_3},\ldots,z_l^{n_l})}{z_2^{n_2}}
			\label{eq::1-sys} \quad \text{for all $i \in \widetilde{I}_-$ and  $j \in \widetilde{I}_+$}
			\\
			z_1^{n_1} > z_3^{n_3}
			\label{eq::2-sys}
			\\
			z_3^{n_3}> \cdots > z_l^{n_l}
			\label{eq::3-sys}
			\\
			h_i(z_3^{n_3},\ldots,z_l^{n_l}) > 0 \quad \text{for all $i \in \widetilde{J}$}
			\label{eq::4-sys}
			\\
			n_2-n_3 > N.
			\label{eq::5-sys}
		\end{numcases}
		Note that as $N\ge0$ and $z_2 = z_3$, the condition~\eqref{eq::5-sys} implies $z_2^{n_2} > z_3^{n_3}$.
		It remains to show how to solve the system (\ref{eq::1-sys}-\ref{eq::5-sys}).

		\emph{Case 1.} Suppose $a_- = a_+ = a > 0$ for some $a \in \rat$.
		This is the only difficult case.
		Recalling that $z_2 = z_3 = \beta$, \eqref{eq::1-sys} is equivalent to
		\[
		\frac{h_i(z_3^{n_3},\ldots,z_l^{n_l})}{z_3^{n_3}} \cdot \frac{1}{\beta^{n_2-n_3}} < \frac{\alpha^{n_1}}{\beta^{n_2}} - a < \frac{h_j(z_3^{n_3},\ldots,z_l^{n_l})}{z_3^{n_3}}  \cdot \frac{1}{\beta^{n_2-n_3}} \quad \text{for all $i \in \widetilde{I}_-$ and $j \in \widetilde{I}_+$.}
		\]
         Observe that $h_i(z_3^{n_3},\ldots,z_l^{n_l}) < h_j(z_3^{n_3},\ldots,z_l^{n_l})$ is implied by~\eqref{eq::1-sys}.
		Inductively solve the system consisting of the inequalities $h_i(z_3^{n_3},\ldots,z_l^{n_l}) < h_j(z_3^{n_3},\ldots,z_l^{n_l})$ for $i \in \widetilde{I}_-$ and $j \in \widetilde{I}_+$, \eqref{eq::3-sys}, and \eqref{eq::4-sys}.
		If no solution exists, then the system (\ref{eq::1-sys}-\ref{eq::5-sys}) does not have a solution either.
		Otherwise, let $(m_3,\ldots,m_l)$ be a solution to the smaller system.
		We will argue that the system (\ref{eq::1-sys}-\ref{eq::5-sys}) also has a solution.

		Define $x_- = \max_{i \in \widetilde{I}_-}\big\{\frac{h_i(z_3^{m_3},\ldots,z_l^{m_l})}{z_3^{m_3}}\big\}$, $x_+ = \min_{i \in \widetilde{I}_+}\big\{\frac{h_i(z_3^{m_3},\ldots,z_l^{m_l})}{z_3^{m_3}}\big\}$ and $\varepsilon = (x_+ - x_-)/4$.
		From the construction of $m_3,\ldots,m_l$ it follows that $\varepsilon > 0$.
		We will construct a solution $(k_1,\ldots,k_l) \in \nat^l$ to the system (\ref{eq::1-sys}-\ref{eq::5-sys}).
		To do this, it suffices to construct $(k_1,\ldots,k_l) $ satisfying (\ref{eq::2-sys}-\ref{eq::5-sys}) with the following additional properties:
		\begin{itemize}
			\item[(a)] $\frac{x_- + \varepsilon}{\beta^{k_2-k_3}} < \frac{\alpha^{k_1}}{\beta^{k_2}} - a <  \frac{x_+ -\varepsilon}{\beta^{k_2-k_3}}$;
			\item[(b)] $\frac{h_i\big(z_3^{k_3},\ldots,z_l^{k_l}\big)}{z_3^{k_3}} < x_- + \varepsilon$ for all $i \in \widetilde{I}_-$;
			\item[(c)] $\frac{h_i\big(z_3^{k_3},\ldots,z_l^{k_l}\big)}{z_3^{k_3}} > x_+ - \varepsilon$ for all $i \in \widetilde{I}_+$.
		\end{itemize}
		As a sanity check on~(a), observe that for any $d\in \nat$,
		\[
		(x_- + \varepsilon)  \frac{1}{\beta^d} < (x_+-\varepsilon)  \frac{1}{\beta^d}.
		\]
		Next, invoke the Pumping Lemma with $m_1,\ldots,m_l$,  $\varepsilon$ as above and the linear forms
		\begin{itemize}
			\item $-h_i(z_3^{n_3},\ldots,z_l^{n_l}) + (x_- +\varepsilon)z_3^{k_3}$ for all $i \in \widetilde{I}_-$,
   			\item $h_i(z_3^{n_3},\ldots,z_l^{n_l}) - (x_+ -\varepsilon)z_3^{k_3}$ for all $i \in \widetilde{I}_+$,
			\item $h_i(z_3^{n_3},\ldots,z_l^{n_l})$ for all $i \in \widetilde{J}$, and
			\item $z_i^{n_i} - z_{i+1}^{n_{i+1}}$ for $3 \le i \le l - 1$
		\end{itemize}
		to compute $\mu, \delta > 0$.
		We have that any $n_3 > m_3$ satisfying
		\[
		\big|\alpha^{\widetilde{n}}/\beta^{n_3} - \mu\big| < \delta
		\]
		for some $\widetilde{n} \in \nat$
		can be extended to $(n_3,\ldots,n_l) \in \nat^{l-2}$ satisfying (\ref{eq::3-sys}-\ref{eq::4-sys}) and (b-c).

		Let $\Delta = \min \left\{ \frac a 2,  \frac{ a\delta}{2\mu} \right\} > 0$.
		It has the properties that $\Delta < a$ and $\mu\Delta/a  \le \delta/2$.
		We will need the following lemma.
		Intuitively, it will be used to show that we can simultaneously satisfy the Diophantine approximation conditions arising from the above application of the Pumping Lemma and item~(a).
		\begin{lemma}
			\label{thm::simult-diophantine-appr}
			Let $a, \mu, \delta, \Delta$ be as above.
			Given $M \in \nat$, we can compute $d > M$ and $m \in \nat$ with the following property.
			For all $k \ge m$, if there exists $\widetilde{n} \in \nat$ such that
			\[
			\big|\alpha^{\widetilde{n}}/\beta^k - a\big| < \Delta,
			\]
			then there exists $\widehat{n} \in \nat$ such that
			\[
			\big|\alpha^{\widehat{n}}/\beta^{k-d}-\mu\big| < \delta.
			\]
		\end{lemma}
		\begin{proof}
			Let $\xi = \delta/(4a)$.
			Using \Cref{thm::kronecker}, choose $d,m \in \nat$ to have the property that $d > M$ and
			\[
			\big\vert \beta^d/\alpha^m -\mu/a \big\vert < \xi.
			\]
			Suppose $|\alpha^{\widetilde{n}}/\beta^k - a| < \Delta$ for some $\widetilde{n} \ge m$.
			Let $\widehat{n} = \widetilde{n} - m$.
			Then
			\begin{align*}
				\left|\frac{\alpha^{\widehat{n}}}{\beta^{k-d}}-\mu\right|  &=
				\left\vert \frac {\alpha^{\widetilde{n}}} {\beta^k}
				\cdot \frac{\beta^d}{\alpha^m} - \mu
				\right\vert\\
				&=
				\left\vert
				\left(\frac{\alpha^{\widetilde{n}}}{\beta^k}-a\right)\frac{\beta^d}{\alpha^m} + a\left(\frac{\beta^d}{\alpha^m} -\frac{\mu}{a}\right)
				\right\vert\\
				&<
				\Delta(\xi+\mu/a) + a\xi\\
				&<2a\xi + \frac{\mu\Delta}{a}\\
				&\le \delta
			\end{align*}
			where the last two inequalities follow from the facts that $\Delta<a$, $\mu\Delta/a \le \delta/2$, and $a\xi = \delta/4$.
		\end{proof}
		Choose $M$ to be such that every $d>M$ has the following properties.
		\begin{itemize}
			\item[(A)] $d > N$;
			\item[(B)] $|x_- +\varepsilon|/\beta^d, |x_+-\varepsilon|/\beta^d < \Delta$;
			\item[(C)] $x_-+\varepsilon + a\beta^d > 1$.
		\end{itemize}
		Then apply \Cref{thm::simult-diophantine-appr} with this value of $M$ to construct $d$ and $m$.
		We will next construct $(k_1,\ldots,k_l) \in \nat^l$ satisfying (\ref{eq::2-sys}-\ref{eq::5-sys}) and (a-c);
		recall that such $(k_1,\ldots,k_l)$ will also be a solution to~(\ref{eq::1-sys}-\ref{eq::5-sys}).
		First, choose $k_1, k_2$ such that $k_2 > \max \{d,m\}$,
		and
		\[
		\frac{x_- + \varepsilon}{\beta^d} < \frac{\alpha^{k_1}}{\beta^{k_2}} - a < \frac{x_+-\varepsilon}{\beta^d}.
		\]
		By~(B), $|\alpha^{k_1}/\beta^{k_2}-a| < \Delta$.
		Then set $k_3 = k_2 -d$.
		By the construction of $d$ and $m$ via \Cref{thm::simult-diophantine-appr}, and the fact that $k_2 > m$, there exists $\widehat{n}$ such that
		\[
		\big|\alpha^{\widehat{n}}/\beta^{k_2-d}-\mu\big|  = \big|\alpha^{\widehat{n}}/\beta^{k_3}-\mu\big|  < \delta.
		\]
		Hence, by construction of $\mu,\delta$ via the Pumping Lemma, we can extend $k_3$ to $(k_3,\ldots,k_l)$ that satisfy (\ref{eq::3-sys}-\ref{eq::4-sys}) as well as (b-c).
		Inequality \eqref{eq::5-sys} and property~(a) are satisfied by construction.
		It remains to show that \eqref{eq::2-sys} is satisfied.
		By~(a), $\alpha^{k_1} - a\beta^{k_2} > (x_-+\varepsilon) \beta^{k_3}$.
		Hence
		\[
		\alpha^{k_1} > (x_-+\varepsilon) \beta^{k_3} + a\beta^{k_2} = \beta^{k_3}(x_-+\varepsilon + a\beta^d) > \beta^{k_3}.
		\]

		\emph{Case 2.}
		Suppose $a_+ > 0$ and $a_+ > a_-$.
		Let $\varepsilon = (a_+ - \max\{a_-, 0\})/4$.
		Compute $M \ge N$ such that for all $n_2,\ldots,n_l \in \nat$ satisfying $z_2^{n_2} > z_3^{n_3} > \cdots > z_l^{n_l}$ and $n_2 - n_3 > M$, we have that
		\[
		\left\vert
		\frac{h_i(z_3^{n_3},\ldots,z_l^{n_l})}
		{z_2^{n_2}}
		\right\vert
		 < \varepsilon \quad \text{for all $i \in \widetilde{I}_- \cup \widetilde{I}_+$}.
		\]
		Next, inductively solve the sub-system comprising inequalities~\eqref{eq::3-sys} and~\eqref{eq::4-sys}.
		If there is no solution, then~(\ref{eq::1-sys}-\ref{eq::5-sys}) does not have a solution and we are done.
		Otherwise, let $(k_3,\ldots,k_l)$ be a solution of the sub-system.
		Applying \Cref{thm::kronecker}, construct $k_1,k_2 \in \nat$ such that $z_1^{k_1} > z_3^{k_3}$, $k_2 - k_3 > M$, and $z_1^{k_1}/z_2^{k_2} \in (a_-+\varepsilon, a_+-\varepsilon)$.
		Then $(k_1,\ldots,k_l)$ is a solution of (\ref{eq::1-sys}-\ref{eq::5-sys}).

		\emph{Case 3.} Suppose $a_+ < a_-$.
		Let $\varepsilon$, $M$, $(k_3,\ldots,k_l)$ be as in Case~2; If no $(k_3,\ldots,k_l)$ exist, once again we are done.
		Observe that any $(n_1,\ldots,n_l) \in \nat^l$ such that $n_2 - n_3 > M$ is not a solution of (\ref{eq::1-sys}-\ref{eq::5-sys}).
		Hence the system (\ref{eq::1-sys}-\ref{eq::5-sys}) has a solution if and only if the system comprising (\ref{eq::1-sys}-\ref{eq::4-sys}) and $N < n_2-n_3 \le M$ has a solution, which can be checked using \Cref{thm::eliminating-one-var}.

		\emph{Case 4.} $a_+ = a_- = 0$.
		In this case, \eqref{eq::1-sys} is equivalent to
		\begin{equation}
			\label{eq::helper-1848}
			h_i(z_3^{n_3},\ldots,z_l^{n_l}) < z_1^{n_1} < h_j(z_3^{n_3},\ldots,z_l^{n_l}) \quad \text{for all $i \in \widetilde{I}_-$ and  $j \in \widetilde{I}_+$}
		\end{equation}
		in which the variable $n_2$ does not appear.
		Hence we can first inductively solve the sub-system comprising~(\ref{eq::1-sys}-\ref{eq::4-sys}).
		If a solution exists, then choose $n_2$ to be sufficiently large to satisfy~\eqref{eq::5-sys}.
		Otherwise, conclude that (\ref{eq::1-sys}-\ref{eq::5-sys}) does not have a solution either.

		\emph{Case 5.} Suppose $a_+ = 0 > a_-$.
		This case is similar to Case 4.
		Let $M$ be such that for all $(n_1,\ldots,n_l)$, if $n_2 - n_3 > M$ then
		\[
		a_- + \frac{h_i(z_3^{n_3},\ldots,z_l^{n_l})}
		{z_2^{n_2}} < 0 \quad \text{for all $i \in \widetilde{I}_-$}.
		\]
		Hence for such $(n_1,\ldots,n_l)$, \eqref{eq::1-sys} is equivalent to
		\begin{equation}
			\label{eq::helper-2359}
			z_1^{n_1} < h_i(z_3^{n_3},\ldots,z_l^{n_l}) \quad \text{for all $i \in \widetilde{I}_+$}.
		\end{equation}
		We therefore solve two systems.
		First, inductively check if the system comprising (\ref{eq::2-sys}-\ref{eq::4-sys}) and \eqref{eq::helper-2359} has a solution $(k_1,k_3,\ldots,k_l)$.
		If yes, then choose $k_2$ to be sufficiently large so that~\eqref{eq::5-sys} is satisfied.
		Thereafter, solve (\ref{eq::1-sys}-\ref{eq::5-sys}) together with $N < n_2 -n_3 \le M$ using \Cref{thm::eliminating-one-var}.
		The system~(\ref{eq::1-sys}-\ref{eq::5-sys}) has a solution if and only if at least one of the two systems has a solution.

		\emph{Case 6.} Finally, suppose, $a_+ < 0$.
		Let $M$ be such that
		\[
		a_+ + \frac{h_i(z_3^{n_3},\ldots,z_l^{n_l})}
		{z_2^{n_2}} < 0 \quad \text{for all $i \in \widetilde{I}_+$}
		\]
		for all $n_2 - n_3 > M$.
		It remains to solve (\ref{eq::1-sys}-\ref{eq::5-sys}) together with $N < n_2 -n_3 \le M$ using \Cref{thm::eliminating-one-var}.

	\end{proof}
	We can finally prove decidability of Problem~\ref{problem1}.
	\begin{proof}[Proof of \Cref{thm::problem1-decidable}]
		\label{proof-of-problem1-decifable}
		We proceed by induction on $l$.
		If $l =1$, then the result is immediate.
		Suppose $l \ge 2$.
		Write the system $A \mathbf{z} > b \:\land\: C\mathbf{z} = d$ in the form
		\[
		\bigvee_{k \in K} A_k\mathbf{z} > \mathbf{b}_k \:\land\: C_k\mathbf{z} = \mathbf{d}_k
		\]
		where each $\mathbf{b}_i \ge \zerovec$.
		By \Cref{thm::pumping-corollary}, this is equivalent to the system
		\[
		\bigvee_{k \in K} A_k\mathbf{z} > \zerovec \:\land\: C_k\mathbf{z} = \mathbf{d}_k.
		\]
		It suffices to solve each disjunct separately.
		Fix $k \in K$.
		If $C_k$ is empty, then we can solve $A_k\mathbf{z} > 0$ using \Cref{thm::ineq-pure}.
		Suppose $C_k$ is non-empty.
		Then first solve $C_k \mathbf{z} = \mathbf{d}_k$  and write the set of solutions $\Scal$ in the form
			\begin{equation*}
			\Scal = \bigcup_{i \in I} \: \bigcap_{j \in J_i} X_{j}
		\end{equation*}
		as in \Cref{thm::equalities-main}.
		It suffices to check, for every $i \in I$, whether $A_k\mathbf{z} > \zerovec$ has a solution belonging to $\bigcap_{j \in J_i} X_{j}$.
		Fix $1 \le i \le I$.
		If $J_i$ is empty, then we simply solve $A_k\mathbf{z} > \zerovec$ using \Cref{thm::ineq-pure}.
		In case $J_i$ is non-empty, we will carry out a variable elimination as follows.
		\begin{lemma}\label{lem: substitution}
			Suppose we are given $\alpha,\beta \in \nat_{>1}$, $z_1,\ldots,z_l \in \{\alpha,\beta\}$, $E \in \intg^{r \times l}$, $\mathbf{u} \in \intg^r$, and $X_1,\ldots,X_M \subseteq \nat^l$ where each $X_j$
			is defined by either
			\begin{equation}\label{eq::helper-1226}
				n_{\mu(j)} = n_{\sigma(j)} + c_{j}
			\end{equation}
			or
			\begin{equation}
				\label{eq::helper-1227}
				n_{\xi(j)} = b_{j}
			\end{equation}
			for some $b_{j}, c_{j} \in \nat$ and $1 \le \xi(j), \mu(j), \sigma(j) \le l$ satisfying $z_{\mu(j)} = z_{\sigma(j)}$.
			We can construct $\lambda < l$, $F \in \intg^{r \times \lambda}$, $\mathbf{v} \in \intg^{r}$, and $y_1,\ldots,y_{\lambda} \in \{\alpha,\beta\}$ such that
			\begin{equation}
				\label{eq::helper-1551}
				E \cdot (z_1^{n_1},\ldots,z_l^{n_l}) > \mathbf{u} \:\land\: (n_1,\ldots,n_l) \in \bigcap_{1 \le j \le M} X_{j}
			\end{equation}
			has a solution if and only if $F \cdot (y_1^{n_1},\ldots,y_\lambda^{n_\lambda}) > \mathbf{v}$ has a solution.
		\end{lemma}
		\begin{proof}
			We proceed by induction on $M$.
			Write $\mathbf{z} = (z_1^{n_1},\ldots,z_l^{n_l})$.
			If $M = 1$, then simply substitute the equation defining $X_1$ into $E \mathbf{z} > \mathbf{u}$.
			Suppose $M > 1$, and consider $X_1$.
			Let
			\[
			(y_1,\ldots,y_{l-1}) = (z_1,\ldots,z_{\xi(j)-1},z_{\xi(j)+1},\ldots,z_l)
			\]
			if $X_1$ is defined by~\eqref{eq::helper-1226}, and
			\[
			(y_1,\ldots,y_{l-1}) = (z_1,\ldots,z_{\mu(j)-1},z_{\sigma(j)+1},\ldots,z_l)
			\]
			if it is defined by~\eqref{eq::helper-1226}.
			Substitute the equation defining $X_1$ into $E \mathbf{z} > \mathbf{u}$ to obtain an equivalent system $\widetilde{E} \mathbf{y} > \widetilde{\mathbf{u}}$.
			Next, for each $k \in \{1,\ldots,M\} \setminus \{\xi(j)\}$, compute an equation of the form~\eqref{eq::helper-1226} or~\eqref{eq::helper-1227} defining $Y_k \subseteq \nat^{l-1}$ that is the projection of $X_k \cap X_1$ onto the appropriate $l-1$ variables.
			If $Y_k$ is empty, then~\eqref{eq::helper-1551} does not have a solution, and we can output any system in $\lambda < l$ variables that does not have a solution.
			If every $Y_k$ is defined by a consistent equation, then invoke the induction hypothesis with the system $\widetilde{E} \mathbf{y} > \widetilde{\mathbf{u}}$ and the sets $Y_k$ for $k \in \{1,\ldots,M\} \setminus \{\xi(j)\}$.
		\end{proof}
		Using the lemma above, we can construct $\lambda < l$, $F \in \intg^{r \times \lambda}$, $\mathbf{v} \in \intg^{r}$, and $y_1,\ldots,y_{\lambda} \in \{\alpha,\beta\}$ such that
		\[
		A_k\mathbf{z} > 0 \:\land\: (n_1,\ldots,n_l) \in \bigcap_{j \in J_i} X_{j}
		\]
		has a solution if and only if
		\begin{equation}
			\label{eq::helper-1626}
			F \cdot (y_1^{n_1},\ldots,y_\lambda^{n_\lambda}) > \mathbf{v}
		\end{equation}
		has a solution.
		Since $\lambda <l$, we can use the induction hypothesis to solve~\eqref{eq::helper-1626}.
	\end{proof}

	\section{Hardness results for the existential fragment of $\mathcal{PA}(\alpha^x, \beta^x)$}
	\label{sec::hardness}


	We now consider the existential fragment of $\mathcal{PA}(\alpha^x, \beta^x)$ for multiplicatively independent $\alpha$ and~$\beta$.
	Unlike the case for $\mathcal{PA}(\alpha^{\nat}, \beta^{\nat})$, we show that decidability of the existential fragment of $\mathcal{PA}(\alpha^x, \beta^x)$ would give us algorithms for deciding various properties of base-$\alpha$ and base-$\beta$ expansions of a large class of numbers, captured by the next definition.

	\begin{definition}
		\label{def::definable}
		A sequence $(u_n)_{n=0}^\infty$ over $\nat$ is existentially definable if for every $k \ge1$ there exists an existential formula $\varphi$ with $k+1$ free variables in the language of $\langle \nat; 0, 1, <, +, x \to \alpha^x, x\to\beta^x \rangle$ such that for all $n,y_0,\ldots, y_{k-1} \in \mathbb{N}$, $\varphi(n,y_0,\ldots,y_{k-1})$ holds if and only if
		\[
		u_{n+i} = y_i
		\]
		for all $0\le i < k$.
	\end{definition}
	The set of definable sequences is closed under many operations.
	Let $(u_n)_{n=0}^\infty$ and $(v_n)_{n=0}^\infty$ be definable, and $c \in \mathbb{N}$. Then $(u_n + v_n)_{n=0}^\infty$, $(c + u_n)_{n=0}^\infty$, $(c \cdot  u_n)_{n=0}^\infty$, $(\alpha^{u_n})_{n=0}^\infty$, $(\beta^{u_n})_{n=0}^\infty$, $(u_{v_n})_{n=0}^\infty$ are also definable.
	Write $\{x\}$ for the fractional part of $x$.
	Let $\alpha,\beta \in \nat_{>1}$ be multiplicatively independent, $(A_n)_{n=0}^\infty$ be the base-$\alpha$ expansion of $\{\log_\beta(\alpha)\}$, and $(B_n)_{n=0}^\infty$ be the base-$\beta$ expansion of $\{\log_\alpha(\beta)\}$.
	Note that $\log_\alpha(\beta), \log_\beta(\alpha)$ are both irrational, and for any $\gamma \in \nat_{> 0}$ and $x \in \rel_{\ge 0}$, the base-$\gamma$ expansions of $x$ and $\{x\}$ differ only by a finite prefix.

	\begin{proposition}
		\label{prop::definability}
		The sequences $(A_n)_{n=0}^\infty$ and $(B_n)_{n=0}^\infty$ are definable.
	\end{proposition}
	\begin{proof}
		By symmetry, it is sufficient to prove the proposition for $(A_n)_{n=0}^\infty$.
		For $x \ge 1$, denote by $f(x)$ the integer $\alpha^m$ such that $\alpha^m\le x  < \alpha^{m+1}$, noting that $f(x) = \alpha ^ {\lfloor x \log_\alpha \beta \rfloor}$.
		Fix $k \ge 1$, and let $w \in \{0,\ldots,\alpha-1\}^k$.
		Denote by $\lambda(w)$ the natural number whose base-$\alpha$ expansion equals $w$.
		That $w$ occurs at position $n$ in $\seq{A_n}$ can be expressed as
		\[
		\lambda(w) < \{ \alpha^n \log_\alpha \beta \} \cdot \alpha^k < \lambda(w)+ 1
		\]
		which is equivalent to
		\[
		\alpha^{\lambda(w)} < \left( \frac{\beta^{\alpha^n}}{\alpha^{\lfloor \alpha^{n} \log_\alpha \beta \rfloor}}\right)^{\alpha^k} < \alpha^{\lambda(w)+1}.
		\]
		Recall that 
		\[
		\alpha^{\lfloor \alpha^{n} \log_\alpha \beta \rfloor} = f(\beta^{\alpha^n}),
		\]
		and for any constant $c$ and a term $t$, we can express $c \cdot t$ as $\underbrace{t + \cdots + t}_{c \textrm{ times}}$.
		Hence the formulas
		\begin{align*}
			\varphi(n,y_0,\ldots,y_{k-1}) \coloneqq
			\exists m \st & \alpha^m \le \beta^{\alpha^n} < \alpha^{m+1} \:\:\land \\
			& \alpha^{\lambda(y_0\cdots y_{k-1}) + m \alpha^k} <\beta^{\alpha^{n+k}} < \alpha^{\lambda(y_0\cdots y_{k-1})+ 1 + m \alpha^k}
		\end{align*}
		for $k \ge 1$ define $\seq{A_n}$ as required.
	\end{proof}
	Observe that we can express whether a pattern $w = w_0\cdots w_{k-1}$ occurs in an existentially definable sequence $\seq{u_n}$ using the existential formula $\exists n\st \varphi(n,w_0,\ldots,w_{k-1})$, where $\varphi$ is the formula described in \Cref{def::definable}.
	Therefore, decidability of the existential fragment of $\mathcal{PA}(\alpha^x,\beta^x)$ would entail existence of oracles, among others, for deciding the following problems.
	\begin{enumerate}
		\item[(A)] Whether a given pattern $w$ appears in the base-$\beta$ expansion of $\log_\beta(\alpha)$.
		\item[(B)] Whether a given pattern $w$ appears at some index simultaneously in the base-$\beta$ expansions of $\log_\beta(\alpha)$ and $\log_\alpha(\beta)$.
		\item[(C)] Whether a given pattern $w$ appears in $(A_{\alpha^n})_{n=0}^\infty$.
	\end{enumerate}
	This proves \Cref{thm::main-hardness} from the Introduction.

	To the best of our knowledge, for no base $\gamma\in \nat_{\ge2}$ and multiplicatively independent $\alpha, \beta \in \nat_{>1}$, an algorithm is known that decides appearance of a given pattern in base-$\gamma$ expansion of $\log_\alpha(\beta)$.
	It is, however, generally believed that $\log_\alpha(\beta)$ is \emph{normal} in every base $\gamma$--that is, every finite pattern $w \in \{0,\ldots,\gamma-1\}$ of length $k$ occurs with frequency $1/\gamma^k$ as a factor in the base-$\gamma$ expansion of $\log_\alpha(\beta)$.\@
	See, for example, \cite[Introduction]{bailey2001}.\@
	For a general exposition to normal numbers, we suggest the reference \cite{bugeaud-book}.\@
	Proof of normality for the sequences $\seq{A_n}$ and $\seq{B_n}$ would make Problem~(A) above trivially decidable.
	However, normality alone is not strong enough to deal with Problems~(B) and~(C):
	deciding the latter problems in the same way as Problem~(A) would require a far stronger ``randomness'' property.
	Even if such properties are proven, we might still be unable to prove decidability of the full existential fragment of $\mathcal{PA}(\alpha^x,\beta^x)$.

	\section{Undecidability of $\mathcal{PA}(\alpha^\nat,\beta^\nat)$}
	\label{sec::undec}
	In this section, let $\alpha,\beta \in \nat_{>1}$ be multiplicatively independent.
	In \cite{hieronymi-schulz}, Hieronymi and Schulz show that the full theory $\mathcal{PA}(\alpha^\nat,\beta^\nat)$ is undecidable by giving a reduction from the Halting Problem for Turing machines.
	We now give an alternative (and shorter) undecidability proof by reducing from the Halting Problem for 2-counter Minksy machines, which is also undecidable \cite[Chapter~14]{minsky1967computation}.\@
	Our proof shows that already for formulas containing three alternating blocks of quantifiers,  membership in $\mathcal{PA}(\alpha^\nat,\beta^\nat)$ is undecidable.

	A \emph{2-counter Minsky machine} is given by $R > 0$ instructions, numbered $1,\ldots,R$, and two counters $c^{(1)}, c^{(2)}$ that take values in $\nat$.
	Each instruction except the $R$th one is either of the form $c^{(i)}  = c^{(i)}  + 1; \:\mathsf{GOTO}\: r$, or $\mathsf{IF} \: c^{(i)} = 0 \: \mathsf{GOTO} \: r \: \mathsf{ELSE \: c^{(i)} = c_i^{(i)} - 1; \: GOTO} \:\widetilde{r}$ where $i \in \{1,2\}$ and $r,\widetilde{r} \in \{1,\ldots, R\}$.
	The execution starts at line $r = 1$ with both counters set to zero, and halts if the line $r = R$ is reached.
	Denote by $c_i^{(n)}$ the value of the counter $c_i$ and by $r_n$ the current instruction number after $n$ steps.
	We refer to $(c^{(1)}_n, c^{(2)}_n, r_n)$ as the \emph{configuration} of the machine at time $n$.
	The \emph{transition function} $f \st \nat\times\nat\times \{1,\ldots,R\} \to \nat\times\nat\times \{1,\ldots,R\}$ of the machine describes how the configuration is updated.
	By definition, we have that $c^{(1)}_0 = c^{(2)}_0 = 0$ and $r_0 = 1$.

    We will represent the trace of the machine by the sequence
	\[
	\langle \alpha^{R+c^{(1)}_0}, \alpha^{R+c^{(2)}_0}, \alpha^{r_0-1}, \alpha^{R+c^{(1)}_1}, \alpha^{R+c^{(2)}_1}, \alpha^{r_1 - 1}, \ldots \rangle.
	\]
	Here, $\alpha^{R+c^{(1)}_n}$ and $\alpha^{R+c^{(2)}_n}$ are at least $\alpha^R$ while $\alpha^{r_n - 1} < \alpha^R$ for every $n \ge 0$.
	Note that every entry in the sequence is a power of $\alpha$, and the $n$th entry is smaller than $\alpha^R$ if and only if $n \equiv 2 \mod 3$.\@
	It remains to represent such sequences using arithmetic of powers of $\alpha$ and $\beta$.

	For $x \in\nat$, denote by $\mu(x)$ the most significant digit in the base-$\alpha$ expansion of $x$, and by $\delta(x)$ the number $\alpha^n$ (whenever it exists) such that the digit corresponding to $\alpha^n$ in the base-$\alpha$ expansion of $x$ is the second most significant digit that is non-zero.
    For example, if $\alpha = 10$, then $\mu(3078) = 3$ and $\delta(3078) = 10^1$.\@
	Next, consider $\Acal_l,\Acal_u \in \alpha^\nat, \Bcal_l,\Bcal_u \in \beta^\nat$ with $\Acal_l < \Acal_u$ and $\Bcal_l < \Bcal_u$.
	Let $\Pcal$ be the set of all $b \in \beta^\nat \cap [\Bcal_l,\Bcal_u]$ such that $\mu(b) = 1$ and $\delta(b) \in [\Acal_l, \Acal_u]$.
	Write $N = |\Pcal| - 1$, and order the elements of $\Pcal$ as $B_0 < \cdots < B_N$.
	We say that the tuple $(\Acal_l, \Acal_u, \Bcal_l, \Bcal_u)$ \emph{defines} the finite sequence $(u_n)_{n=0}^N$ over $\alpha^\nat$ given by $u_n = \delta(B_n)/\Acal_l$.
	The following result is Lemma 3.4 in \cite{hieronymi-schulz}, and serves a crucial role in their and our undecidability proofs.

	\begin{theorem}
		\label{thm::hs-black-box}
		Every finite sequence $(u_n)_{n=0}^N$ over $\alpha^\nat$ is defined by some $(\Acal_l, \Acal_u, \Bcal_l, \Bcal_u)$.
	\end{theorem}
	By choosing $\Bcal_l$ to be the smallest element of $\Pcal$ and $\Bcal_u$ to be the largest element of $\Pcal$ if necessary, we can always assume that $\Bcal_l, \Bcal_u \in \Pcal$.
	We will encode the Halting Problem for 2-counter machines by constructing a formula that expresses existence of a tuple $(\Acal_1, \Acal_2, \Bcal_1, \Bcal_2)$ that defines a sequence corresponding to a finite trace of the machine ending with the halting instruction.
	Let $\Acal_l,\Acal_u\in \alpha^\nat, \Bcal_l,\Bcal_u \in \beta^\nat$ define the sequence $(u_n)_{n=0}^N$, and $\Pcal = \{B_0, \ldots, B_N\}$ be as above.
	Define
	\begin{align*}
		\varphi_{\Acal_l,\Acal_u,\Bcal_l,\Bcal_u}\,(C,A,B) \coloneqq \,&
		C\in \alpha^\nat \:\land\:A \in \alpha^\nat \cap [\Acal_l, \Acal_u] \:\land\: B \in \beta^\nat \cap [\Bcal_l, \Bcal_u] \:\land\:
		\\
		&C \le B < 2C \:\land\:
		 A \le B-C < \alpha \cdot A.
	\end{align*}
	This formula states that $B \in \Pcal$, which is witnessed by $C$ and $A$.
	Here, $C$  is the largest power of $\alpha$ not exceeding $B$, the atomic formula $C \le B < 2C$ ensures that $\mu(B) = 1$, and $A \le B-C < \alpha \cdot A$ ensures that $A = \delta(B)$.
	If $\varphi_{\Acal_l,\Acal_u,\Bcal_l,\Bcal_u}\,(C,A,B)$ holds, then $u_n = A/\Acal_l$ where $n$ is the position of $B$ in $\Pcal$.
	The next formula, on input $B_1, B_2$ that belong to $\Pcal$, returns whether $B_1$ is immediately followed by $B_2$  in the ordering of $\Pcal$.
	\begin{align*}
		\psi_{\Acal_l,\Acal_u,\Bcal_l,\Bcal_u}(B_1,B_2) \coloneqq
		\forall C, A, B_1 < B < B_2  \st \neg\varphi_{\Acal_l,\Acal_u,\Bcal_l,\Bcal_u}(C,A,B).
	\end{align*}
	We omit the subscript from $\phi$ and $\psi$ when $\Acal_l,\Acal_u,\Bcal_l,\Bcal_u$ are clear from the context.
	We can now construct a formula in the language $\Lcal_{\alpha,\beta}$ that is true if and only if the given 2-counter machine halts.
	Write $\mathbf{X}$ for the collection of variables $\Acal_l,\Acal_u,\Bcal_l,\Bcal_u, \widehat{B}_1,\widehat{B}_2, \widehat{C}_0,\widehat{C}_1,\widehat{C}_2, C_{\text{last}}$, and $\mathbf{Y}$ for the collection of variables $C_0,A_0,B_0,\ldots,C_5,A_5,B_5$.
	The variables
	\begin{itemize}
		\item $\Acal_l,\Acal_u,\Bcal_l,\Bcal_u$ serve to define a finite sequence over $\alpha^\nat$,
		\item $\Bcal_l, \widehat{B}_1,\widehat{B}_2$ denote the first three elements of $\Pcal$ with witnesses $(\widehat{C}_0,\alpha^R\cdot\Acal_l)$, $(\widehat{C}_1, \alpha^R\cdot\Acal_l)$, and $(\widehat{C}_2, \Acal_l)$, respectively,
		\item $\Bcal_u$ is the final element of $\Pcal$ with the witness $(\widehat{C}_{\text{last}},\alpha^{R-1} \cdot \Acal_l)$, and
		\item $C_0,A_0,B_0, \ldots, C_5,A_5,B_5$ represent arbitrary~6 consecutive terms of the sequence defined by $(\Acal_l,\Acal_u,\Bcal_l,\Bcal_u)$, which correspond to two consecutive configurations of the machine. (Recall that each configuration of the machine consists of three numbers.)
	\end{itemize}
	The required formula is then
	\begin{align*}
		\exists \mathbf{X} \st
		&\psi(\Bcal_l, \widehat{B}_1) \:\land\: \psi(\widehat{B}_1, \widehat{B}_2) \:\land\: \varphi(\widehat{C}_0, \alpha^R\cdot\Acal_l, \Bcal_l)\:\land\: \varphi(\widehat{C}_1, \alpha^R \cdot \Acal_l, \widehat{B}_1) \:\land\: \varphi(\widehat{C}_2,\Acal_l, \widehat{B}_2) \:\land\:
		\\
		&\varphi(C_{\text{last}}, \Bcal_u, \alpha^{R-1} \cdot \Acal_l) \:\land\:\\
		& \forall \mathbf{Y} \st \bigg(
		\bigwedge_{i=0}^4 \psi(B_i,B_{i+1})
		\:\land\:
		\bigwedge_{i=0}^5 \varphi(C_i,A_i,B_i)
		\:\land\:
		A_2 < \alpha^R \cdot \Acal_l
		\bigg) \Rightarrow \Phi(C_0,A_0,B_0,\ldots,C_5,A_5,B_5)
	\end{align*}
	where $\Phi$ implements the transition function of the machine.
	Note that $\Acal_l,\Acal_u,\Bcal_l,\Bcal_u$ also appear in the definitions of $\varphi$ and $\psi$.\@
	The first row in the formula above
	fixes the initial configuration of the machine to $(0, 0, 1)$ by requiring that the first three elements of the sequence defined by $(\Acal_l,\Acal_u,\Bcal_l,\Bcal_u)$ must be $\alpha^R, \alpha^R, 1$, respectively.
	The second row says that the last term in the sequence must be $\alpha^{R-1}$, which represents the halting instruction.
	The condition $A_2 < \alpha^R \cdot \Acal_l$ in the third row, in conjunction with $\varphi(C_2,A_2,B_2)$, ensures that the term of the sequence at the position defined by $B_2$ represents an instruction number, as opposed to a counter value.
	Thus $(C_0,A_0,B_0), \ldots,(C_5,A_5,B_5)$ represent two consecutive configurations of the machine.
	{Regarding~$\Phi$,} observe that we can define a function mapping $\alpha^n $ to $\alpha^{n+1}$ (which corresponds to incrementing a counter) by the formula $\chi(x,y) \coloneqq y = \underbrace{x + \cdots + x}_{\textrm{$\alpha$ times}}$, and a function mapping $\alpha^{n+1}$ to $\alpha^n$ (corresponding to decrementing a counter) by $\widetilde{\chi}(x,y) \coloneqq \chi(y,x)$.
	Finally, to see that the formula above has quantifier alternation depth 2 (i.e.\ three alternating blocks of quantifiers), recall that $\chi_1 \Rightarrow \chi_2$ is equivalent to $\lnot \chi_1 \lor \chi_2$ and the definition of $\psi$ involves a single universal quantifier.

	\paragraph*{Acknowledgements.}
	Toghrul Karimov and Jo\"el Ouaknine were supported
	by the DFG grant 389792660 as part of TRR 248 (see
	\url{https://perspicuous-computing.science}). Jo\"el Ouaknine is also
	affiliated with Keble College, Oxford as \texttt{emmy.network} Fellow.
	James Worrell was supported by EPSRC Fellowship EP/X033813/1.
	The Max Planck Institute for Software Systems is part of the Saarland
	Informatics Campus.
	The authors thank the anonymous referees for their constructive comments, and Mihir Vahanwala and Valérie Berthé for helpful initial discussions.
	
    \bibliographystyle{plain}
    \bibliography{main}

\end{document}